\documentclass[10pt,twocolumn,twoside]{IEEEtran}

\usepackage{generic}
\usepackage{amsmath,amssymb,amsfonts}
\usepackage{graphicx}
\usepackage{textcomp}
\usepackage{ulem}

\def\BibTeX{{\rm B\kern-.05em{\sc i\kern-.025em b}\kern-.08em
    T\kern-.1667em\lower.7ex\hbox{E}\kern-.125emX}}
    
\usepackage{paralist}
\usepackage{graphics} 
\usepackage{epsfig} 
\usepackage{epstopdf}
\usepackage{xcolor}
\usepackage{bbm}
\usepackage{tikz}
\usetikzlibrary{arrows}
\usepackage{bm}
\usepackage{circuitikz}

\tikzstyle{block} = [draw, fill=blue!20, rectangle, 
    minimum height=3em, minimum width=6em]
\tikzstyle{input} = [coordinate]
\tikzstyle{sum}=[draw,minimum size=1mm,inner sep=0pt,outer sep=0pt,shape=circle,fill=black]
\tikzstyle{output} = [coordinate]
\tikzstyle{pinstyle} = [pin edge={to-,thin,black}]
\tikzstyle{int}=[draw, fill=blue!20, minimum size=3em]
\tikzstyle{init} = [pin edge={to-,thin,black}]

\usepackage{subcaption,graphicx}

\usepackage{pgf}
\usepackage[labelsep=period]{caption}
\usepackage{hyperref}
\usepackage{url}
\usepackage{setspace,booktabs}
\usepackage{caption}
\usepackage{wrapfig}
\usepackage{lipsum}%
\usepackage{comment}
\usepackage{textpos}

\usepackage[symbol]{footmisc}
\renewcommand{\thefootnote}{\fnsymbol{footnote}}

\newtheorem{thm}{Theorem}
\newtheorem{cor}{Corollary}

\newtheorem{assum}{Assumption}
\newtheorem{prop}{Proposition}

\newtheorem{exmp}{Example}

\DeclareMathOperator{\R}{\mathbb{R}}

\begin{document}

\title{\LARGE \bf Risk of Phase Incoherence in Wide Area Control of \\ Synchronous Power Networks}
\author{Christoforos Somarakis, Guangyi Liu  and Nader Motee 
\thanks{*Christoforos Somarakis is with Palo Alto Research Center, a Xerox Company, 3333 Coyote Hill Road
Palo Alto, CA 94304 USA {\tt\small somarakis@parc.com}, Guangyi Liu,  and Nader Motee are with the Department of Mechanical Engineering and Mechanics, Lehigh University, Bethlehem, PA, 18015, USA.  \{\tt\small gliu,motee\}@lehigh.edu}}

\maketitle

\begin{abstract}
We develop a framework to quantify systemic risk measures in a class of Wide-Area-Control (WAC) laws in power networks in the presence of noisy and time-delayed sensory data.  A closed-form calculation of the risk of phase incoherence in interconnected power networks is presented, and the effect of network parameters, information flow in WAC architecture, statistics of noise, and time-delays are characterized. We show that in the presence of time-delay and noise, a fundamental trade-off between the best achievable performance (via tuning feedback gains) and value-at-risk emerges. The significance of our results is that they provide a guideline for developing algorithmic design tools to enhance the coherency and robustness of closed-loop power networks simultaneously. Finally, we validate our theoretical findings through extensive simulation examples.
\end{abstract}

\section{Introduction}

The modern power networks have been struggling with evermore narrow stability margins under the ongoing deregulation of energy markets and the growth of highly volatile sources of renewable energy as well as variable load endpoints  \cite{fox2010smart,momoh2012smart}. These systems may be steered to undesirable contingency events under continued stressful operating conditions. The 1996 Western American blackout \cite{blackout1996}, the 2003 blackouts in North-East USA, Canada \cite{blackout2003_2} and Italy \cite{blackoutitaly_1} are real-world examples of failures due to outdated stabilization techniques or naive interconnectivity \cite{1705631,5454394}.


The synchronous power systems are considered robust if they remain in or return to synchronism after experiencing a fault event. The concept of transient stability measures this property, where it characterizes the extent to which generators remain in phase as they recover from a nontrivial disturbance \cite{sauer2006power,1457634}. The transient dynamics in power systems are handled via local power system stabilizer (PSS) modules. These are controllers that perform feedback stabilization in the event of generator excitation \cite{1994power}. However, PSSs perform poorly in the absence of coordinating authority that may result in grid instability \cite{260906}. Furthermore, the existing proposed solutions require careful tuning with high-gain feedback control laws \cite{1709097} that suffer from scalability issues \cite{6740090}. A different approach utilizes Wide-Area-Control (WAC) methods \cite{6580901,6740090,5728885}, where network stabilization relies on remote measurements from the entire grid. These methods enjoy the advantage of being technologically feasible due to high-bandwidth phasor measurement units and flexible AC transmission system devices \cite{doi:10.1002/etep.545}. Furthermore, the theory of multi-agent systems is reaching the level of maturity to provide scalable algorithms for control of modern grids \cite{Annaswamy2013,bullo2018lectures,6740090}.
On the downside, measurements have to be transmitted over some cyber layer \cite{SOUDBAKHSH2017171}. Therefore, in contrast to local control, WAC is vulnerable to asynchronous propagation of information or measurement noise that corrupts actual data. Moreover, their simultaneous interplay may severely impact WAC capabilities on stabilizing power networks \cite{6740090}.   In addition to the aforementioned references, some other related works investigate stability problems in power systems with nonlinear models \cite{SCHIFFER2017261,1457634} and time-delay mitigation in WAC loops \cite{Vu2017,985276,5688286,6517501}. The latter work focuses on the worst-case scenarios for robust control, which usually yields conservative control policies.




In this work, we consider the problem of phase incoherence in linear transient dynamics of a noisy power network.  We assess the performance of a specific class of WAC policies that seeks to control the power grid through a virtual communication network. The network is prone to asynchronous processing and propagation of information. Moreover, information transmitted between stations is corrupted by measurement noise that captures the impact of imperfect sensors and/or communication protocols.  A schematic diagram of the problem setup is illustrated in Figure \ref{fig: problem}. The robustness of the closed-loop network is assessed through the notion of  Value-At-Risk, which is a systemic risk measure and adopted from Finance literature \cite{follmer11} and has been recently utilized by the authors in the context of multi-agent control protocols \cite{somarakisnader_tac_1}-\cite{8884747} as a surrogate for robustness that scales with high-dimensional dynamical systems.

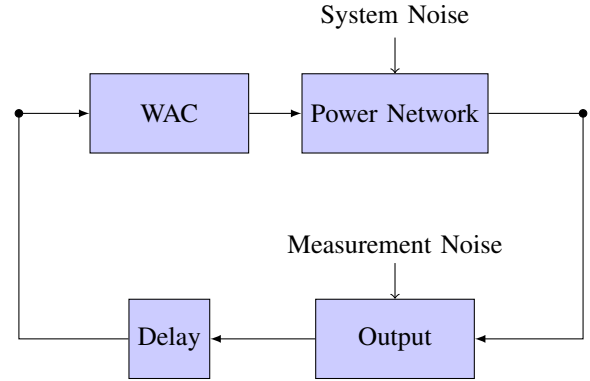
\begin{figure}
\center
\begin{tikzpicture}[auto, node distance=2cm,>=latex]
	\node [input, name=input] {};
	\node [sum, right of=input] (sum) {};
	\node [block, right of=sum] (controller) {WAC};
	\node [block, right of=controller, pin={[pinstyle]above:System Noise}, node distance=3cm] (system) {Power Network};
	    \node[sum, right of=system, node distance=2.5cm] (y) {};
	\node [block, below of=system,pin={[pinstyle]above: Measurement Noise}, node distance=3cm] (measurements) {Output};
	\node [int, below of=controller, node distance=3cm] (delay) {Delay};
	\draw [->] (controller) -- node[name=u] {} (system);
	\draw [->] (sum) -- node {} (controller);
	\draw [->] (measurements) -- (delay);
	\draw [-] (system) -- (y);
	\draw [->] (y) |- (measurements);
	\draw [-] (delay) -| (sum);
\end{tikzpicture}     \caption{A schematic diagram of WAC of power networks with heterogeneous sources of noise is depicted. In addition, data transmitted over the communication network is  noisy and time-delayed.}\label{fig: problem}
\end{figure}

We examine the role of WAC in reducing the risk of undesirable behaviors in power networks in the presence of measurement noise and time-delay. We introduce a notion of nested systemic events to characterize various levels of experiencing undesirable events. The value-at-risk of phase incoherence between a given pair of synchronous generators is calculated using the corresponding nested systemic events. It is shown that the systemic risk measure depends on the spectral properties of the underlying graph of the network, the WAC gains, statistics of noise, and time-delay. In the presence of time-delay, fundamental limits emerge between the best (lowest) achievable levels of risk and highest possible values for the feedback gains in the WAC. This limit suggests that while higher feedback gains may result in better performance, e.g., smaller total resistive power loss \cite{6759860}-\cite{7086037}, it will increase the risk of phase incoherence in the network. This paper is an outgrowth of  \cite{SOMARAKIS2018142,somnader20}, which contains several new technical contributions with respect to its conference versions. 

\section{Mathematical Notations}\label{sect: notation} 

We denote the non-negative orthant of the Euclidean space $\R^n$ by and  $\mathbb{R}_{+}^n$. The $n\times n$ diagonal matrix is denoted by $D=\text{diag}\big \{d^{(i)}\big\}$, the $n \times r$ zero matrix is $O_{n\times r}$, and $O_n$ when $r=n$. The $n\times n$ identity matrix is denoted by $I_n$. Finally, by $1_n$ we understand the $n\times 1$ vector of all ones.

\vspace{0.1cm}
{\it Algebraic Graph Theory:} A undirected weighted graph is defined by $\mathcal{G} = (\mathcal{V}, \mathcal{E}, \omega)$, where $\mathcal{V}$ is the set of nodes, $\mathcal{E}$ is the set of edges (feedback links), and $\omega: \mathcal{V} \times \mathcal{V} \rightarrow \mathbb{R}_{+}$ is the weight function that assigns a non-negative number (feedback gain) to every link. Two nodes are directly connected if and only if $(i,j) \in \mathcal{E}$.

\begin{assum}  \label{asp:connected}
    Every graph in this paper is connected. In addition, for every $i,j \in \mathcal{V}$, the following properties hold:
    \begin{itemize}
        \item $\omega(i,j) > 0$ if and only if $(i,j) \in \mathcal{E}$.
        \item $\omega(i,j) = \omega(j,i)$, i.e., links are undirected.
        \item $\omega(i,i) = 0$, i.e., links are simple.
    \end{itemize}
\end{assum}

The Laplacian matrix of $\mathcal{G}$ is a $n \times n$ matrix $L = [l_{ij}]$ with elements
\[
    l_{ij}:=\begin{cases}
        \; -\omega(i,j)  &\text{if } \; i \neq j \\
        \; \omega(i,1) + \ldots + \omega(i,n)  &\text{if } \; i = j 
    \end{cases},
\]
where $k_{i,j} := \omega(i,j)$. Laplacian matrix of a graph is symmetric and positive semi-definite. Assumption \ref{asp:connected} implies the smallest Laplacian eigenvalue is zero with algebraic multiplicity one. The spectrum of $L$ can be ordered as 
$
    0 = \lambda_1 < \lambda_2 \leq \dots \leq \lambda_n.
$
The eigenvector of $L$ corresponding to $\lambda_k$ is denoted by $\bm{q}_{k}$. By letting $Q = [\bm{q}_{1} | \dots | \bm{q}_{n}]$, it follows that $L = Q \Lambda Q^T$ with $\Lambda = \text{diag}[0, \lambda_2, \dots, \lambda_n]$. We normalize the Laplacian eigenvectors such that $Q$ becomes an orthogonal matrix, i.e., $Q^T Q = Q Q^T = I_{n}$ with $\bm{q}_1 = \frac{1}{\sqrt{n}} \bm{1}_n$. The complete incidence matrix with dimensions of $n(n-1)/2\times n$ is denoted by $B_n$.

\vspace{0.1cm}
{\it Probability Theory:} Let $\mathcal{L}^{2}(\mathbb{R}^{q})$ be the set of all $\R^q$-valued random vectors $z = [z^{(1)}, \dots ,z^{(q)}]^T$ of a probability space $(\Omega, \mathcal{F}, \mathbb{P})$ with finite second moments. A normal random variable $y \in \mathbb{R}^{q}$ with mean $\mu \in \mathbb{R}^{q}$ and $q \times q$ covariance matrix $\Sigma$ is represented by $y \sim \mathcal{N}(\mu, \Sigma)$. We employ standard notation $\text{d} \bm{\xi}_t$ for the formulation of stochastic differential equations. Finally, with $\nu_{\varepsilon}$ we denote the unique solution of $\int_{-\nu_\varepsilon}^{\nu_\varepsilon}e^{-\frac{t^2}{2}}\,dt=\sqrt{2\pi}(1-\varepsilon)$ for $\varepsilon \in (0,1)$.

\section{Time-Delayed Synchronous Power Networks}\label{sect: formulation} 
Let us consider a network of $n$ synchronous generators connected over $m$ transmission lines. The $i$'{th} generator is defined through the (static) triplet $(J_i,\beta_i,E_i)$ and (dynamic) state vector $(\theta^{(i)}_t,\omega^{(i)}_t)$. The triplet of fixed parameters consists of the rotational inertia, the damping coefficient, and the voltage magnitude, respectively. The state variables at time instant $t$ are the rotor phase $\theta^{(i)}_t$ and the rotor frequency $\omega^{(i)}_t=\frac{d}{dt}\theta^{(i)}_t$. The states of the synchronous generators can be stacked as $$\theta_t=\big[\theta_t^{(1)},\dots,\theta_t^{(n)}\big]^T~~~~\text{and}~~~~\omega_t=\big[\omega_t^{(1)},\dots,\omega_t^{(n)} \big]^T,$$ to define the network's state vector. 
Interactions in the power network are characterized through the reduced complex admittance matrix $Y=[Y_{ij}\angle \varphi_{ij}]$, where the angles $\varphi_{ij}\in [0,\frac{\pi}{2}]$ are the phase shifts that characterize the energy loss due to the transfer conductance. Throughout the paper, any transfer conductance is assumed zero, i.e., $\varphi_{ij}=\frac{\pi}{2}$ for all $i\neq j$. This makes $Y_{ij}$ the normalized susceptance of the transmission line connecting the $i^{th}$ and $j^{th}$ generators. The effective power input $p_i$ of generator $i$ is defined as the difference of the mechanical power and the electrical internal voltage power, i.e., $p_i=P_i^{m} - E_i^2 \, \mathrm{Re} \{Y_{ii}\}$.

We consider the following benchmark model \cite{1457634} of the dynamic phase of the $i$'{th} generator 
\begin{equation}\label{eq: classicalmodel}
	J_i\ddot{\theta}^{(i)}_t=-\beta_i \dot{\theta}^{(i)}_t+\sum_{j=1}^n E_i E_j Y_{ij} \sin(\theta^{(j)}_t-\theta^{(i)}_t)+p_i
\end{equation} 
for $i=1,\dots,n$. Although this model has been criticised for being over-simplified  \cite{7924418}, it has proven to be a reliable and benchmark model to investigate control policies in power networks \cite{Vu2017}. 
The central objective of transient stability is to investigate conditions under which \eqref{eq: classicalmodel} will converge to an operating equilibrium point. For fixed voltage magnitudes, admittances, and power inputs, the equilibrium point belongs to the manifold 
\begin{equation*}
	\begin{split}
		\mathbb S=\bigg\{&\big(\theta,~ \omega\big) \in \mathbb R^{2n}~\Big|~ \omega=0 ~~\text{and} ~~\big|\theta^{(i)}-\theta^{(j)}\big|<\frac{\pi}{2}~~\text{with} \\
 		&p_i=\sum_{j=1}^n E_i E_j Y_{ij} \sin\big(\theta^{(i)}-\theta^{(j)}\big),~i,j=1,\dots,n \bigg\} 
	\end{split}
\end{equation*} 
The condition that geodesic distance between phase angles is restricted in $(-\pi/2,\pi/2)$ allows for the consistency of angle differences;  we refer to \cite{doi:10.1137/110851584} for more details.


\subsection{Deviation from synchronous states} Let us denote the desired equilibrium point by  $(\theta_*,0)\in \mathbb S$. Due to exogenous stochastic disturbances, the phase and frequency states tend to fluctuate around the equilibrium state. This allows us to investigate behaviour of the system \eqref{eq: classicalmodel} in presence of additive noise using its linearization around the equilibrium point\footnote{To keep our notations simple, the state variables $\theta_t$ and $\omega_t$ have been recycled to denote the error variables $\theta_t-\theta_*$ and $\omega_t-\omega_*$.} by considering the error dynamics      
\begin{equation}\label{eq: linearized3}
	\begin{split}
		\begin{bmatrix}
		d\theta_t \\ d\omega_t
		\end{bmatrix}&=\left[
	\begin{array}{ccc}
		0  & \vline &  I \\ \hline
		-L & \vline & -D
	\end{array}\right] \begin{bmatrix}
		\theta_t \\
		\omega_t
		\end{bmatrix}dt+ \begin{bmatrix}
		0 \\ 
		H
	\end{bmatrix} d\xi_t,
	\end{split}
\end{equation}  
where $d\xi_t=\big[d\xi^{(1)}_t,\dots,d\xi^{(n)}_t\big]^T$. 
The corresponding Laplacian matrix $L=[l_{ij}]$ is defined elementwise as 
$$ 
l_{ij}=\begin{cases}
J_{i}^{-1}E_iE_jY_{ij}\cos(\theta^{(i)}_*-\theta^{(j)}_*) &~~~\textrm{if}~~~ i \neq j \\
-  \displaystyle \sum_{k\neq i}l_{ik} &~~~\textrm{if}~~~i=j
\end{cases}
$$ 
 and the diagonal matrix $D=\text{diag}\{\frac{\beta_1}{J_1},\dots,\frac{\beta_n}{J_n}\}$ contains the damping over inertia ratios of power machines. In this model,  the stochastic fluctuations in effective power inputs $p_i$ for $i=1,\ldots, n$ are considered as one source of exogenous disturbance, where the effects of mismatch between power generation and load demand are modeled as multiple mutually independent white noise processes $d\xi^{(1)}_t,d\xi^{(2)}_t\dots,d\xi^{(n)}_t$. The diffusion coefficient for the $i$'th generator is $\eta J_i^{-1}$. Therefore, the diffusion matrix is given by $H=\eta \, J^{-1}$, where $J=\mathrm{diag}\{J_1, \ldots, J_n\}$.





%
%

\subsection{Virtual  State Feedback  Control}\label{subsect: control} 

In \eqref{eq: linearized3}, the term involving $d\xi_t$ induces stochastic fluctuation for the state vector around $(\theta_*,0) \in \mathbb S$. One possible approach to mitigate the effect of the exogenous disturbance is to adjust power voltage on admittance parameters, i.e., to probe into elements of $L$ and $D$. However, this is a costly, if not infeasible, solution. For large-scale power networks, the distributed WAC with the aid of a communication network of sensors/controllers is a tractable alternative. The idea is for actuators to be installed on the virtual network and use measurements of their adjacent remote peers \cite{6740090,5728885}. 
To this end, we consider introducing the state feedback control in \eqref{eq: linearized3}, such that
  \begin{equation}\label{eq: linearizedwithfeedback}
		\begin{bmatrix}
		d\theta_t \\ d\omega_t
		\end{bmatrix}=\left[
	\begin{array}{ccc}
		0  & \vline &  I \\ \hline
		-L & \vline & -D
	\end{array}\right] \begin{bmatrix}
		\theta_t \\
		\omega_t
		\end{bmatrix}dt+ \begin{bmatrix}
		0 \\ 
		H
	\end{bmatrix} d\xi_t +\begin{bmatrix}
0 \\  I
\end{bmatrix} u_t,
\end{equation} 
where the control input $u_t$ utilizes information of the virtual network from the power network. The communicating sensors are assumed to operate in parallel with the electric grid and established over multiple spatially distributed local controllers. 
As is the case with every transmitted and processed information, time-delays in signals are ubiquitous and they account for severe impacts on the closed-loop system.
\begin{assum}\label{assum: delay}
	The control mechanism over the communication network processes state information with a constant and uniform time-delay $\tau>0$.
\end{assum} 

The main sources of time-lagged information are generic latency in measurements, communication, and actuation of the control policies. 
In this work, we regard $\tau>0$ as a nominal constant that all time-delay sources are lumped into \footnote{In reality, time-delay occurs for various reasons, and its amount ought to be considered stochastic, time-dependent, or heterogeneous, which is beyond the scope of our work.}.

The virtual state feedback control for generator $i$ is given by  
\begin{align}\label{eq: feedback}
        u_t^{(i)} & =-\sum_{j=1}^n \Big[ m_{ij} \left(\theta_{t-\tau}^{(j)} +\eta'  d\xi_t^{(j+n)} \right) \nonumber \\
        & \hspace{3cm} +  k_{ij} \left(\omega_{t-\tau}^{(j)} +\eta'      d\xi_t^{(j+2n)} \right) \Big]. 
\end{align}
The real data are usually corrupted by measurement noises due to, e.g., imperfect sensors, simultaneous casting, or other conditions that may occur at other levels of communication \cite{Stallings:2006:DCC:1215307}. The effect of such uncertainties on the rotor phase and frequency in the virtual feedback control law are modeled as $\eta' \, d\xi_t^{(j+n)}$ and  $\eta' \, d\xi_t^{(j+2n)}$ , for $j=1,\dots,n$, respectively, where the diffusion coefficient $\eta'$  models the noise magnitude.

The feedback gains, which characterize interactions among the  communicating controllers  over the virtual network, are assumed to be symmetric, i.e., $m_{ij}=m_{ji}$ and $k_{ij}=k_{ji}$. By denoting the corresponding interaction matrices by $M=[m_{ij}]$ and $K=[k_{ij}]$,  the overall closed-loop network can be represented in the following compact form 
\begin{equation}\label{eq: model}
	\begin{split}
		\begin{bmatrix}
		d\,\theta_t \\ d\,\omega_t
		\end{bmatrix}&=\mathbf{A} \begin{bmatrix}
		\theta_t \\
		\omega_t
		\end{bmatrix}dt+\mathbf{K} \begin{bmatrix}
		\theta_{t-\tau} \\
		\omega_{t-\tau}
		\end{bmatrix}dt+ \mathbf{H}\, d\xi_t
	\end{split}
\end{equation} 
in which 
\[
\mathbf{A}=\left[
	\begin{array}{ccc}
		0  & \vline &  I \\ \hline
		-L & \vline & -D
	\end{array}\right], ~\mathbf K=\left[
\begin{array}{ccc}
0 & \vline & 0  \\ \hline
-J^{-1}M & \vline & -J^{-1}K
\end{array}\right].
\]
Summing up perturbation input and communication noise, the diffusion matrix  can be written as 
\begin{equation}\label{eq: H}
\mathbf{H} =\left[ \begin{array}{ccccc}
      0 & \vline & 0 & \vline &0 \\ \hline
    \eta J^{-1}  & \vline & \eta' M & \vline & \eta' K 
\end{array}\right] \in \R^{2n \times 3n}.
\end{equation} 
The diffusion matrix includes two types of uncertainties, i.e., the exogenous noise and the measurement noise, that affect the dynamics of the network simultaneously, but independently. The network \eqref{eq: model} is initialized with arbitrary, but fixed, functions $\phi(t)=\big(\phi^{\theta}(t),\omega^{\theta}(t)\big)$ over time interval $t\in [-\tau,0]$ that are statistically independent of $d\xi_0$.

\begin{assum}\label{assum: uniformity}
For all $i=1,\dots,n$, the damping and inertia parameters of generators are assumed to be identical, i.e., $$\beta_i = \beta>0  ~~~~\text{and}~~~~ J_i = J>0.$$
\end{assum}

Although this simplification implies that all generators have the same damping over inertia ratios, the above assumption has been widely addressed in the literature \cite{DBLP:journals/corr/abs-1711-10348,doi:10.1137/110851584,7924418,Siami16TACa} and will facilitate the risk analysis and enable the possibilities to obtain explicit results. 

\subsection{Problem Statement}      
In the error dynamics \eqref{eq: model}, the states of all generators in the power network will fluctuate around the operating equilibrium point. Hence, some of the generators may experience phase incoherence and deterioration in robustness. Our goal is to quantify value-at-risk of phase incoherence in the power network as a function of the underlying graph Laplacian, communication time-delay in virtual control, and noise statistics. Furthermore, we  characterize fundamental limits and tradeoffs that emerge from existence of time-delay and noise in the system.

\section{Internal Stability \& Measurement Statistics}\label{sect: preliminary} 

In order to calculate systemic risk measures, one needs to develop some intermediate results about stability and statistics of the steady state variables, which constitutes the subject of this section. First, we establish explicit necessary and sufficient conditions for asymptotic stability of the unperturbed system, i.e., in \eqref{eq: model} when $\eta=\eta'=0$, where the stability region depends on system and feedback control parameters and the time-delay  $\tau>0$. Then, based on the derived stability window, we will calculate the steady-state statistics of the network. 

\subsection{Simultaneous Diagonalizable Feedback Control}
To allow derivation of analytical results, the feedback gain matrices, which are of interest in this work, satisfy the following diagonalizability condition. 
\begin{assum}\label{assum: commute} 
		The feedback gain matrices $M$ and $K$ are designed such that each pair out of $L,M,K$ commutes.
\end{assum}  
 
A standard result  \cite{Horn:2012:MA:2422911} asserts that this assumption is equivalent to the existence of a unitary matrix $Q$ such that $Q^T  U Q$ is diagonal for every $U \in \{L,M, K \}$ (cf. Theorem 4.1.6 of \cite{Horn:2012:MA:2422911}). Therefore, by virtue of Assumption \ref{assum: commute}, it is assumed that there exists matrix $Q$ such that $Q^{T}LQ=\Lambda_L,~Q^{T}MQ=\Lambda_{M}$  and $Q^{T}KQ=\Lambda_{K}$, where  $\Lambda_{L},~\Lambda_{M},~\Lambda_{K}$ are  diagonal matrices. In particular, columns of $Q$ can be arranged so that $\Lambda_L=\text{diag}\big\{\lambda_1,\lambda_2,\dots,\lambda_n \big\}$ with elements sorted as $0=\lambda_1<\lambda_2\leq \dots  \leq \lambda_n.$  Similarly, the $j$'th diagonal elements of $\Lambda_{M}$ and $\Lambda_{K}$ are $\mu_j$ and $\kappa_j$, respectively.

Examples of simultaneously diagonalizable structures with respect to $L$ appear when $M$ and $K$ are proportional to: identity matrix $I_n$, graph Laplacian $L$, centering matrix $M_n$, or matrix $B_nB_n^T$, where $B_n$ is the incidence matrix. Another case is when $M$ and $K$ are proportional to each other and  $M$ commutes with $L$.  While Assumption \ref{assum: commute} may impose some restrictions on admissible WAC policies, one can still form a feasible WAC policy. The reason is that feedback gains $M$ and $K$ are only parameters of the virtual network for which it is reasonable to  assume that the designer has enough synthesis authority to choose admissible gains, whereas the power network matrices $L,~D$ are much less accessible to any redesign and modification. 

\subsection{Internal Stability}



For general matrices $M$ and $K$, sufficient conditions for the asymptotic behavior of the unperturbed closed-loop system can be inferred using standard methods in the stability theory of delay differential equations \cite{gu2003stability}. These conditions are usually conservative and of little use for further analysis. However, when $M$ and $K$ satisfy Assumption \ref{assum: commute}, one can explicitly characterize the stability region of the unperturbed system. 
For the exposition of the following result, we introduce the sets $\big\{\mathbb W_i(s;k)~\big|~(s;k)\in \mathbb R^{2}_+ \times \mathbb R^2 \big\}_{i=0}^3$ that are defined in Table \ref{table: stability} of Appendix \ref{appendix1}. 
\begin{thm}\label{thm: stability} 
    Suppose that Assumptions \ref{assum: uniformity}-\ref{assum: commute} hold. The solution of \eqref{eq: linearizedwithfeedback} when noise is absent and time-delay $\tau>0$ converges to $\big[ \rho {\bf 1}_n , 0_n\big]^T $ for some $\rho\in \mathbb R$ if and only if $$\big(d\tau ,\lambda_j\tau^2~;~\mu_j\tau^2,\kappa_j\tau \big) \in \bigcup_{r=0}^3 \mathbb W_r$$
    for $j=1,\dots,n$, where $d$ is the uniform damping ratio and $\lambda_j,~\mu_j,~\kappa_j$ are the $j$'th eigenvalue of matrices $L$, $M$ and $K$, respectively. Furthermore, if $\mu_1=\kappa_1=0$, then $\rho=\frac{1}{n}1_n^T \phi^\theta(0)+ \frac{1}{dn} 1_n^T\phi^{\omega}(0)$.
\end{thm}

\subsection{Measurement Statistics}
Theorem \ref{thm: stability} describes the parametric set that guarantees the asymptotic convergence of the unperturbed system to $\mathbb S$. The infused Uncertainty into the network generates stochastic fluctuations of the state variables around the equilibrium. It is convenient to study these solutions using  the associated transition matrix $\Phi(t)$\footnote{The solution of zero-input system \eqref{eq: model} are typically represented with the $2n\times 2n$ transition matrix $\Phi(t)$  that solves the matrix-valued system after initialization $\Phi(t)\equiv 0$ for $t<0$ and $\Phi(0)=I_{2n}$  \cite{gu2003stability,lunel93}.} that retains the stability properties of the system according to Theorem \ref{thm: stability}.  The overall process is formulated as the sum of a deterministic and a stochastic term \begin{equation}\label{eq: solution}
    \begin{bmatrix}
    \theta_t \\
    \omega_t
    \end{bmatrix}=T_t+\int_0^t \Phi(t-s)\begin{bmatrix}
    O_{n\times 3n} \\ 
    H
    \end{bmatrix}\,d\xi_s.
\end{equation} The term $T_t=\int_{-\tau}^0\Phi(t-s-\tau)\phi(s)\,d\mu(s)$ is presented as a Stieltjes integral for some appropriate measure $\mu$ with bounded variation; see \cite{Mohammed84}. Equation \eqref{eq: solution} illustrates the statistics of states of the error dynamics. For $t<\infty$, the system's state admits a normal random vector distribution as
$$
    \begin{bmatrix}
    \theta_t \\ \omega_t
    \end{bmatrix} \sim \mathcal N~\bigg( T_t, \int_0^t \Phi(s) \begin{bmatrix}
    O_n & O_n \\ O_n & H\,H^T
    \end{bmatrix} \Phi^T(s)\,ds \bigg ).
$$
The quantity of interest for risk assessment  is the phase incoherence between two generators, i.e., the elements of the vector 
$$
    y_t=B_n\, \theta_t.
$$
One way to obtain insight on the manner with which the power grid, the controllers, the time-delay, and the noise contribute to phase incoherency is to consider the long-term statistics of $y_t$, i.e., $\overline{y}:=\lim_{t \rightarrow +\infty}y_t$. 

\begin{thm}\label{thm: statistics} 
Suppose that Assumptions \ref{assum: delay} and  \ref{assum: commute} hold for the solution $(\theta_t,\omega_t)$ of system \eqref{eq: model} and let us consider function $f:\bigcup_{r=1}^3 \mathbb W_r \rightarrow \mathbb R_+ $ which is defined in Appendix \ref{app: spectralfunctions}. The output vector $y_t$ converges, in distribution, to 
	$$
    	\overline y \sim \mathcal N \bigg( 0, \frac{1}{2\pi  }B_n\, Q \; \textnormal{diag}\big\{\, \mathfrak{f}_l\, \big\}\,Q^T\,  B_n^T \bigg),
	$$ 	
	where $\mathfrak{f}_l=\mathfrak{f}_l\big((s;k)_l\big)$ for every  $l=1,\dots,n$ is defined by 
	$$\mathfrak{f}_l=\begin{cases}
    	0 & \textrm{if} ~~~~ l = 1 \\
    	\tau^3\left[\frac{\, \eta^2}{J^2}+ \eta'\,^2 \left(\frac{(k_1)_l^2}{\tau^4}+\frac{(k_2)_l^2}{\tau^2} \right)\right]f\big((s;k)_l\big) & \textrm{if} ~~~~ l>1
    	\end{cases},
	$$ and $f=f\big((s;k)_l\big)$ is spectral function defined in \eqref{eq: f} and evaluated at $(s;k)_l:=(s_1,s_2;k_1,k_2)_l$ in which $$(s_1,s_2;k_1,k_2)_l:=\big( d\,\tau,\lambda_l\,\tau^2~;~\mu_l\,\tau^2, \kappa_l \,\tau\big)~\in~\bigcup_{r=1}^3 \mathbb W_r.$$
	
\end{thm}

The spectral function $f(s;k)$ is implicitly defined using an improper integral and its basic properties are discussed in Appendix \ref{app: spectralfunctions}. Using the statistics of $\overline{y}$, we will calculate the risk of phase incoherence in the error dynamics \eqref{eq: model}. 

\section{Risk of Phase Incoherence} \label{sect: main}
The stability of \eqref{eq: classicalmodel} is challenged when the error dynamics exceed certain margins, which may be the case for certain types and magnitudes of $\mathbf H$ in the first-order approximation. Higher-order terms may then prevail and steer the system away from the nominal equilibrium. This section aims to quantify ``how much uncertainty'' can the linearized dynamics sustain before they are steered into unsafe regions, which are named the \textit{systemic sets} henceforth. 

\begin{figure}
    \centering
    \includegraphics[scale=0.55]{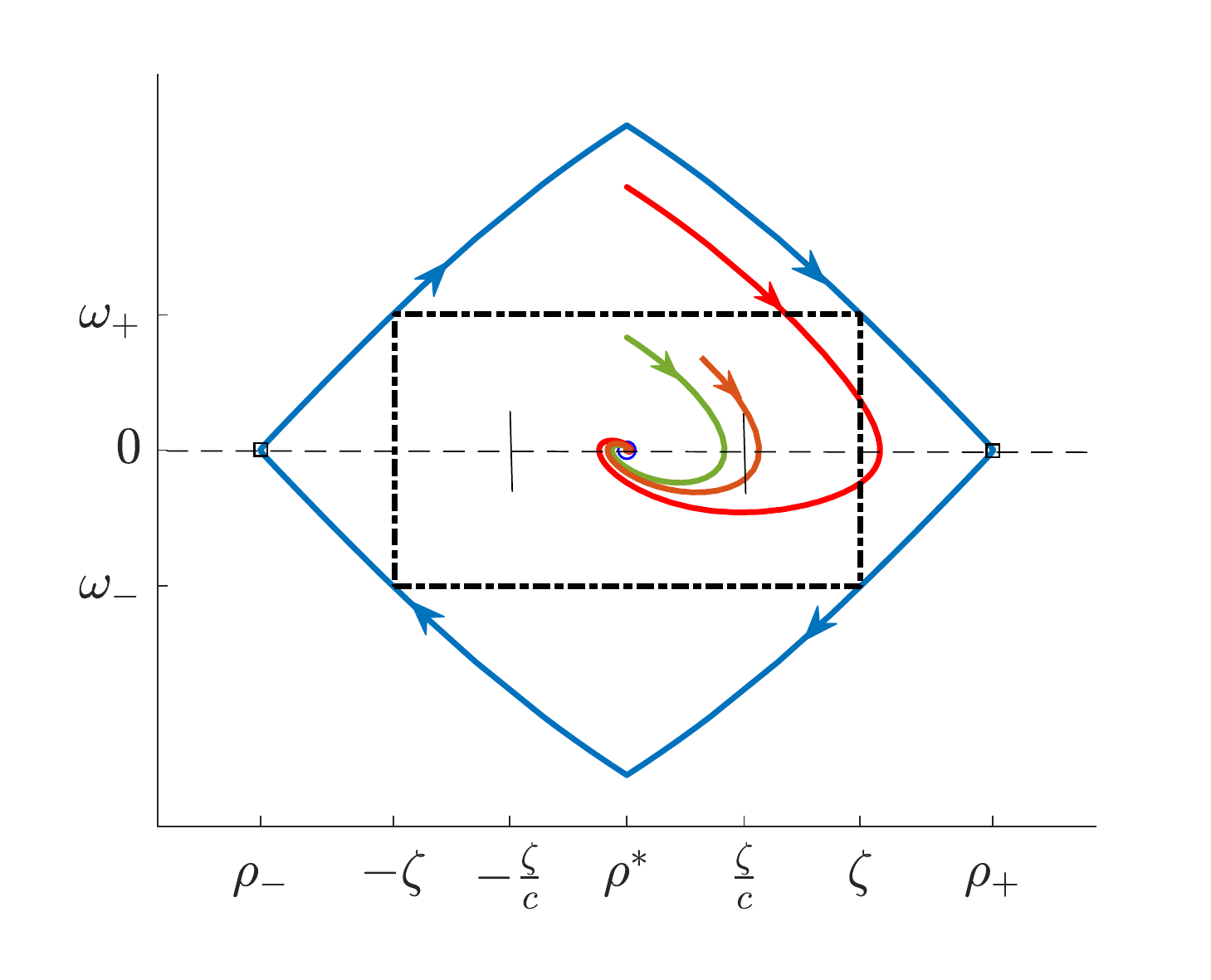}
    \caption{The schematics of systemic sets for phase coherence. The swing system \eqref{eq: classicalmodel} admits countable equilibria both stable and unstable. Parameters $\zeta$ and $c$ can be selected with respect to separatrix curves and fixed bounds on admissible frequencies $(\omega_-, \omega_+)$.}
    \label{fig: schematics}
\end{figure}

\subsection{The Systemic Set} 
The $(i,j)$'th element of the column vector $y_t$ measures the relative coherence of two synchronous generators $i$ and $j$, i.e., $\theta_t^{(j)}-\theta_t^{(i)}$. We consider systemic sets on the range of $|\theta_t^{(j)}-\theta_t^{(i)}|$ as follows. Evidently, the closer the elements of $y_t$ are to zero, the safer the performance of the grid. Let us choose $\zeta>0$ as the lower limit for which $|y_t^{(i,j)}| > \zeta$ should not happen and any value $|y_t^{(i,j)}| \in (\zeta,\infty)$ is flagged as unsafe. An additional checkpoint between $0$ and $\zeta$ is defined as the permissible, presumably harmless, magnitude of fluctuations, i.e., $\zeta/c \in (0,\zeta)$ with $c>1$, which introduces a trichotomy of characteristic ranges for $|y_t^{(i,j)}|$. The underlying assumptions behind this construction are 
\begin{itemize}
    \item When $|y_t^{(i,j)}| \in [0,\zeta/c)$, there is no risk of encountering phase incoherence and the risk value is assigned as zero.
    \item When $|y_t^{(i,j)}| \in [\zeta/c, ~ \zeta]$, it will trigger a positive value of risk, as an index of safety margin between phase incoherence and undesirable behavior.
    \item For consistency reasons, as phase incoherence approaches $\zeta$, the risk increases and becomes $+ \infty$ for any value $|y_t^{(i,j)}| \in (\zeta,\infty)$. 
\end{itemize}

\vspace{2mm}
\begin{exmp}
The motivation behind the $(\zeta,c)$ trichotomy stems from standard transient stability analysis \cite{1457634}. Let us consider  a $2$-machine system with coupled dynamics \eqref{eq: classicalmodel} with nominal (constant) internal voltages $p_1,p_2$ that satisfies the uniformity Assumption \ref{assum: uniformity}. If $J|p_1-p_2|<2E_1 E_2 Y_{12}$, the phase difference dynamics attains a family of equilibria $\big\{0,\theta_1^*-\theta_2^*\big\}$, for which  $\sin(\theta_1^*-\theta_2^*)=\frac{J(p_1-p_2)}{2E_1 E_2 Y_{12}}$, and are of interlacing stability. In Figure \ref{fig: schematics}, we depict three such fixed points $(0,\rho_-)$ (unstable), $(0,\rho_*)$ (stable), $ (0,\rho_+)$ (unstable). Separatrices that connect the unstable equilibria essentially determine the stability region of $(0,\rho_*)$. As explained in \cite{1457634}, the transient stability assesses the ability of post-fault dynamics to remain in the region of attraction of nominal stable equilibrium. Meaningful choices of parameters $\zeta$ and $c$ can relate margins of error dynamics governed by \eqref{eq: linearizedwithfeedback} from areas of the initial phase space that lie close to separatrix curves and stability boundaries. In Figure \ref{fig: schematics}, we additionally assume that the frequency difference ought to remain in $(\omega_-,\omega_+)$, which assists in focusing on phase incoherence only. We can then appropriately select $\zeta$ and $c$ as functions of $\omega_{\pm}$ and separatrix curves. Analysis of frequency incoherence is beyond the scopes of this work. However, in the discussion section, we highlight the necessary steps towards this direction.
\end{exmp}
\vspace{2mm}

By introducing parameter $\delta>0$, the formulation of the aforementioned desired properties for areas of phase incoherence yields the family of nested sets $\{U_\delta \}_{\delta>0}$  where 
$$
    U_\delta =\bigg( \zeta \frac{1+\delta}{c+\delta},+\infty \bigg). 
$$ 
It is immediate to identify the two extremes $U_0:=(\zeta/c,\infty)$ and $U_\infty:=(\zeta,\infty)$ as the safe (zero risk) set and the undesirable set, respectively. 

\subsection{Risk Evaluation}
The distribution of the observables enables us to investigate robustness properties in the noisy power grid with respect to systemic sets. The notion of  value-at-risk measure of systemic events is a suitable index known to provide an answer on how close a random variable may get to an undesirable set in probability \cite{Rockafellar_Royset_2015}, which is defined as  
\begin{equation}\label{eq: riskdef}
\mathcal R_{\varepsilon}(y)=\inf\big\{\delta>0 ~\big|~ \mathbb P \big(|y|~\in~U_\delta\big)<\varepsilon \big\}.
\end{equation}
The design parameter $\varepsilon \in (0,1)$ determines the rigidity of $\mathcal R_\varepsilon$. For the values of $\varepsilon$  close to 1, the index \eqref{eq: riskdef} becomes less reliable; whereas for $\varepsilon$ close to $0$, the index may become overly rigid\footnote{For more details on this family of risk measures, we refer to \cite{follmer11,Rockafellar_Royset_2015} or \cite{somarakisnader_tac_1} for more details and their applications in networked control systems.}. The next result provides a closed form expression on risk in phase incoherence. For its exposition, we recall the root value $\nu_\varepsilon$  of $\int_{\nu_\varepsilon}^{-\nu_\varepsilon} e^{-t^2/2}\,dt=\sqrt{2\pi}(1-\varepsilon)$.

\begin{thm}     \label{thm: risk}
	The risk of steady-state phase incoherence $|\overline{y}^{(i,j)}|=\lim_{ t\rightarrow +\infty} |\theta_t^{(i)}-\theta_t^{(j)}|$ between generators $i$ and $j$ with error dynamics \eqref{eq: model} is:
	\begin{equation}\label{eq: risk}
	\mathcal R_{\varepsilon}^{(i,j)}=
		\begin{cases}
		0 &~~\text{if}~~\sigma_{ij} \leq \frac{\zeta}{c \nu_\varepsilon}\\
		\frac{\sigma_{ij}\,\nu_{\varepsilon}\,c-\zeta}{\zeta-\sigma_{ij}\,\nu_{\varepsilon}} & ~~\text{if}~~ \frac{\zeta}{c \nu_\varepsilon} < \sigma_{ij} < \frac{\zeta}{\nu_\varepsilon} \\
		+\infty & ~~\text{if}~~ \sigma_{ij} \geq \frac{\zeta}{\nu_\varepsilon}
		\end{cases}
	\end{equation} where $\sigma_{ij} =\sqrt{\frac{1}{2\pi}\sum_{l=2}^n (q_{il}-q_{jl})^2\, \mathfrak{f}_l},$ and $\big\{\mathfrak{f}_l\big\}_{l=2}^n $ as in Theorem \ref{thm: statistics}. 
\end{thm} 
The above theorem calculates the value-at-risk of phase incoherence for pairs of power generators. One can also construct the risk profile of the entire network by stacking  the risk values for all pairs  into one vector as 
\begin{equation*}
	\mathfrak R_{\varepsilon}=\big( \,\mathcal R_{\varepsilon}^{(1,2)},\mathcal R_{\varepsilon}^{(1,3)},\dots,\mathcal R_{\varepsilon}^{(n-1,n)} \, \big)^T.
\end{equation*}

\subsection{Special Cases}
We conclude this section by considering the risk in two simplified scenarios for the communication network.

\vspace{2mm}
\subsubsection{Synchronous State-Feedback ($\tau=0$)}
When there is no time-delay in the feedback control, contrary to \eqref{eq: risk}, the risk formula for $\tau=0$ admits a quite simple form that is worth reporting. 

\begin{cor}\label{cor: risk2}
	The risk of steady-state phase incoherence $|\overline{y}^{(i,j)}|=\lim_t |\theta_t^{(i)}-\theta_t^{(j)}|$ between generators $i$ and $j$ with synchronous ($\tau=0$) closed-loop control error dynamics \eqref{eq: model} is given by \eqref{eq: risk}
	with 
	$$\sigma_{ij} =\sqrt{\sum_{l=2}^n (q_{il}-q_{jl})^2\frac{\frac{\eta^2}{J^2}+\eta'^2(\kappa_l^2+\mu_l^2)}{2(d+\kappa_l)(\lambda_l+\mu_l)}}.$$
\end{cor}
The above result illustrates how the risk of systemic events is associated with the power network properties, the two sources of disturbances, and the Laplacian eigenvalues of feedback gain matrices for the simultaneously diagonalizable type of control that concerns the present work.

\vspace{2mm}

\subsubsection{Noise-Free Measurements $(\eta'=0)$}
When we can neglect sensor imperfections, the remaining source of uncertainty is the exogenous input in the power network that models the load volatility, which is parameterized by $\eta$. In this case, one also obtains a simple risk formula that follows directly from Theorems \ref{thm: statistics} and \ref{thm: risk}. 

\begin{cor}\label{cor: risk3}
The risk of steady-state phase incoherence $|\overline{y}^{(i,j)}|=\lim_t |\theta_t^{(i)}-\theta_t^{(j)}|$ between generators $i$ and $j$ with perfect measurement data to control the error dynamics \eqref{eq: model} is given by \eqref{eq: risk}
with 
$$
    \sigma_{ij} =\frac{\tau^{3/2}\eta}{J\sqrt{2\pi}}\sqrt{\sum_{l=2}^n (q_{il}-q_{jl})^2\, f\big((s;k)_l\big)},
$$  
in which the spectral function $f$ is evaluated as in Theorem \ref{thm: statistics}. 
\end{cor}

\section{Fundamental Limits and Trade-Offs}\label{sect: limits} 
The value-at-risk analysis reveals the vulnerability of WAC policies in the presence of time-delay and corrupted measurements. This section explores these limitations separately, establishes universal design bounds in mitigating the risk of power network phase incoherence, and draws remarks on their combinatorial effect. 


\vspace{2mm}

\subsubsection{Delay-Induced Limitations} 
The emergence of delay-induced fundamental lower limits of systemic risk has been previously reported in the context of multi-agent systems in \cite{somarakisnader_tac_1,8884747}. We show that a non-trivial lower bound on the steady-state variance in presence of time-delay and exogenous noise imposes a lower limit on the networked control law to mitigate the risk of systemic events. For the exposition of this result, we assume noiseless measurement, i.e., $\eta'=0$. Our analysis of the spectral function $f$ that is conducted in Appendix \ref{app: spectralfunctions}  helps us  identify the existence of a lower bound for $f(s_l;k)$ for fixed $s=s_l=(d\tau,\lambda_l\tau^2)$, where $\lambda_l$ is the $l$'th eigenvalue of $L$. Then, let us define $$\underline{f_l}:=  \min_{(s_l;k) \in \bigcup_{r=1}^3 \mathbb W_r}\,f\big((s_l;k)\big).$$ 




The following result highlights how pairwise standard deviation is always lower bounded by the load uncertainty parameter and the time-delay.    
\begin{prop}\label{prop: limitsigma} 
	For any pair of generators, the variance of phase differences is lower bounded by 
	$$\sigma_{ij}\geq \sigma_*:=\frac{\tau^{3/2}\eta}{J\sqrt{2\pi}}\min_{(i\neq j)}\sqrt{\, \sum_{l=2}^n (q_{il}-q_{jl})^2\, \underline{f_l}}.$$  
\end{prop}%

\vspace{2mm}

\subsubsection{Noise-Induced Limitations}
The measurement noise also imposes hard limits. This limit is independent of time-delay and can be deduced from the form of $\sigma_{ij}$ in Theorem \ref{cor: risk3}. It can be shown that the pair standard deviation $\sigma_{ij}$ is a function of eigenvalues of feedback matrices $K,M$ and that  $\{\mu_l,\kappa_l\}_{l=2}^n$ can be optimized to a global minimum, which depend on the basic characteristics of the power network and noise parameters. This suggests strong evidence of hard limitation on the feedback control strategy\footnote{For a more technical argument on the matter, the interested reader can use basic calculus and show that functions of the form $\zeta(x,y)=\frac{A^2+B^2(x^2+y^2)}{(d+x)(\lambda +y)}$, that constitute $\sigma_{ij}$ in Corollary \ref{cor: risk2}, always attain a positive global minimum.}.

\vspace{2mm}

The coexistence of both time-delay and measurement noise in our model elevates emergence of  stronger and more complex limitations on WAC policies that can be studied via function $\mathfrak{f}$ introduced in Theorem \ref{thm: statistics}. Such analysis is algebraically tedious and beyond the scope of this work. Nevertheless, it is straightforward to conclude that the existence of minimum values of $\sigma_{ij}$ characterizes limits on risk-aware control design. We provide the following technical result on lower bound of the risk measure for the class of simultaneously diagonalizable feedback controllers.

\begin{thm}\label{thm: limit} 
	Given systemic set parameters $\zeta$, $c$, and the acceptance level $\varepsilon \in (0,1)$, one of the following statements must hold.
	\begin{itemize}
		\item When $\sigma^* \leq \frac{\zeta}{c\,\nu_{\varepsilon}}$, the gain matrices $M$ and $K$ can be tuned to make $\mathfrak R_{\varepsilon}$ arbitrarily small. 
		\item When $\frac{\zeta}{c\,\nu_{\varepsilon}} < \sigma^*<\frac{\zeta}{\nu_{\varepsilon}}$, the least achievable risk value cannot be reduced beyond  
		$$
		\mathfrak R_{\varepsilon} \succeq  \left( \frac{\sigma^* \, \nu_\varepsilon\, c -\zeta}{\zeta-\sigma^*\,\nu_{\varepsilon}} \right) 1_{r},
		$$ 
		for $r=\frac{n(n-1)}{2}$.
		\item  When $\sigma^* \geq \frac{\zeta}{\nu_{\varepsilon}}$,  any choice of $M,K$ results in infinite risk of phase incoherence. 
	\end{itemize}
\end{thm}

\vspace{2mm}

\subsubsection{Limits and Trade-Offs in Consensus WAC}
A particular type of state feedback controller that satisfies Assumption \ref{assum: commute} is when $M=\mu L$ and $K = \kappa L$, where $L$ is the power grid Laplacian matrix and $\kappa>0$, $\mu>0$ are scaling factors. This type of feedback controller obtains a consensus structure since $M$ and $K$ are Laplacian matrices. The sensor network is now a consensus network that applies time-delayed feedback control to the grid. The connectivity of the network can be quantified via the notion of the effective resistance $\Xi$ \cite{klein1993}  and its convenient spectral form that for our initial power network graph reads $\Xi_{L}=\sum_{l=2}^n \lambda_l^{-1}$. It follows that
\begin{equation}
	\Xi_{K}=\sum_{l=2}^n \frac{1}{\kappa_{l}}=\frac{1}{\kappa}\sum_{l=2}^n \frac{1}{\lambda_{l}} = \frac{\mu}{\kappa} \sum_{l=2}^n \frac{1}{\mu_{l}}=\frac{\mu}{\kappa}\, \Xi_{M},
\end{equation} 
where $\Xi_M$ and $\Xi_K$ are the corresponding effective resistances of phase and frequency consensus control of the power network.
It is known that the stronger the inter-connectivity of a graph, the smaller its effective resistance. The stability region $\bigcup_{r=1}^3 \mathbb W_{r}$ introduced in Sec. \ref{sect: preliminary} imposes lower bounds on both $\Xi_{K}$ and $\Xi_{M}$ that we can express as follows. Let $\lambda_n=\|L\|$ be the maximum Laplacian eigenvalue of the underlying graph and define the set $\mathbb Q =\big \{(\mu,\kappa) \in \mathbb R^2_+~|~ (d\tau,\|L\|\tau^2,\|L\|\mu \tau^2,\|L\|\kappa \tau)\in \partial\{\bigcup_{r=1}^3 \mathbb W_{r}\} \big\}$. The effective resistances of $K$ and $M$ are lower bounded by 
\begin{equation}    \label{eq: resistances}
	\begin{split}
		\Xi_{K}&=\sum_{l=2}^n\frac{1}{\kappa_l}>\frac{(n-1)}{\max \big\{ \kappa \hspace{0.05cm} \big| \hspace{0.05cm} (\kappa,\mu)\in \mathbb Q \big\} \|L\|},\\
		\Xi_{M}&=\sum_{l=2}^n\frac{1}{\mu_l}>\frac{(n-1)}{\max \big\{\mu\hspace{0.05cm} \big| \hspace{0.05cm}(\kappa,\mu)\in \mathbb Q\big\} \|L\|}.
	\end{split}
\end{equation} 


Contrary to the lower limits of the systemic risk \cite{somarakisnader_tac_1,8884747}, the above limitations on the graph structure of WACs are only due to the time-delay. At this point, the risk of phase incoherence and the effective resistance are restricted by independent factors.
\begin{thm}\label{thm: trade-off} 
	Given systemic set parameters $\zeta$, $c$, and the acceptance level $\varepsilon \in (0,1)$,  there exists a common limit for the product of the systemic risk and the effective resistance,
	
	\begin{equation}\label{eq: tradeoff}
    	\mathfrak R_{\varepsilon} \cdot \sqrt{\Xi_K+\Xi_M} \succeq \mathrm 1_{r}\, \Omega,
    \end{equation} 
    where $\Omega>0$ is a universal constant depending on the grid properties, time-delay, and uncertainty constants $\eta'$, $\eta$. 
\end{thm}
The above theorem asserts that it is impossible to design a feedback control that simultaneously minimizes the risk of systemic events and maximizes the networked-control connectivity, which are measures by $\Xi_K$ and $\Xi_M$, beyond a specific threshold. It can be shown that the total resistive power loss for the linearized model \ref{eq: model} depends directly to the network parameters as well as $\Xi_K$ and $\Xi_M$ \cite{6759860,7086037}. The discussion of such trade-offs are beyond the scopes of this paper as in this paper we only consider transmission lines with purely imaginary admittance, i.e., zero resistive losses. 








\section{Case Studies}\label{sect: casestudy}
The first example to be studied is a two-machine structure, which is usually addressed in literature as suppressed multi-machine networks or single-machine infinity bus systems \cite{pai1989}. The second case involves risk-aware analysis and synthesis in the IEEE-39 standard.  

\begin{figure*}[t]
    \begin{subfigure}[t]{.33\linewidth}
        \centering
        \includegraphics[width=\linewidth]{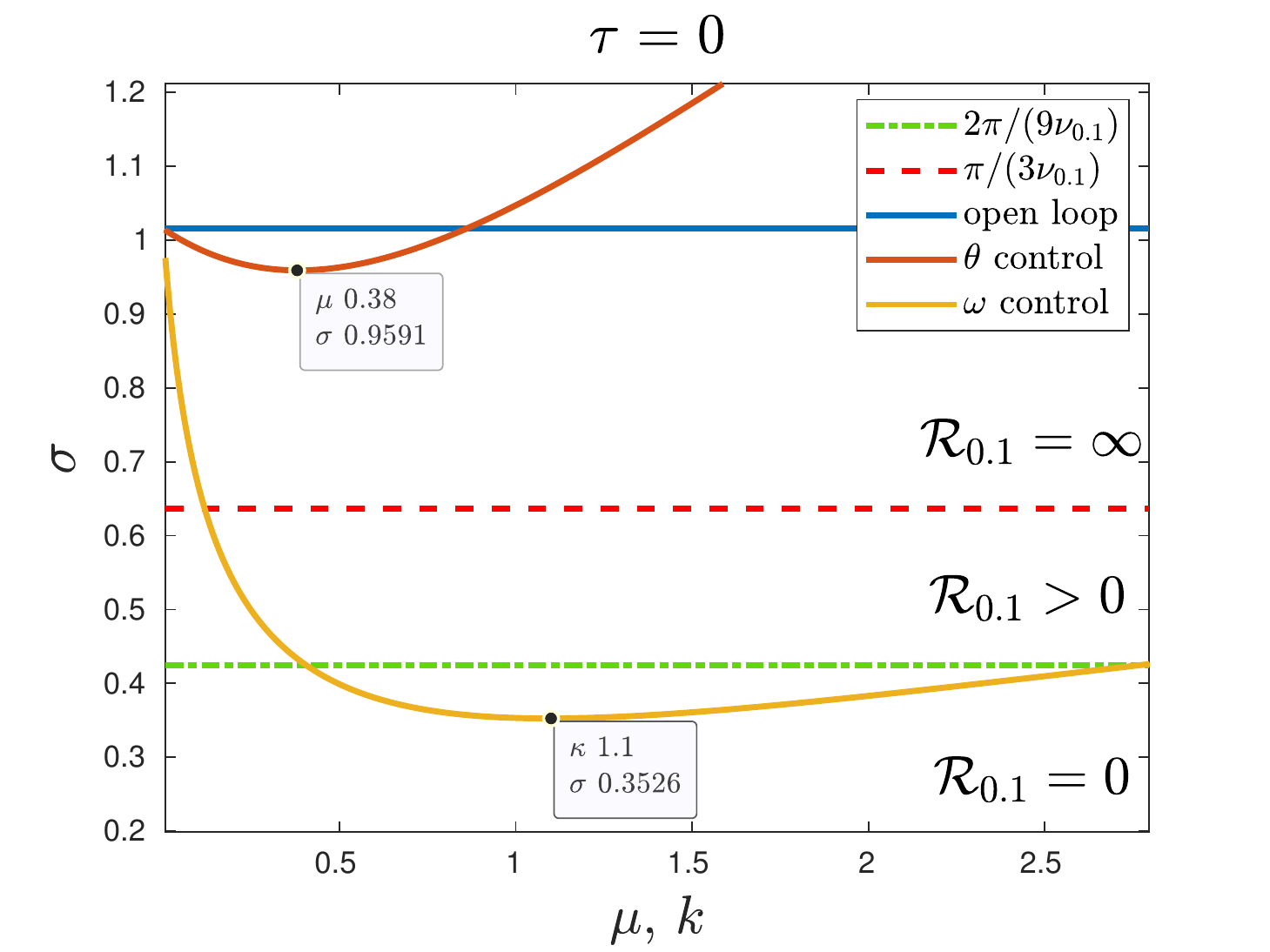}
        \caption{Separate $\omega$ and $\theta$ controls.}
        \label{fig: example1a}
    \end{subfigure}
    \hfill
    \begin{subfigure}[t]{.32\linewidth}
        \centering
        \includegraphics[width=\linewidth]{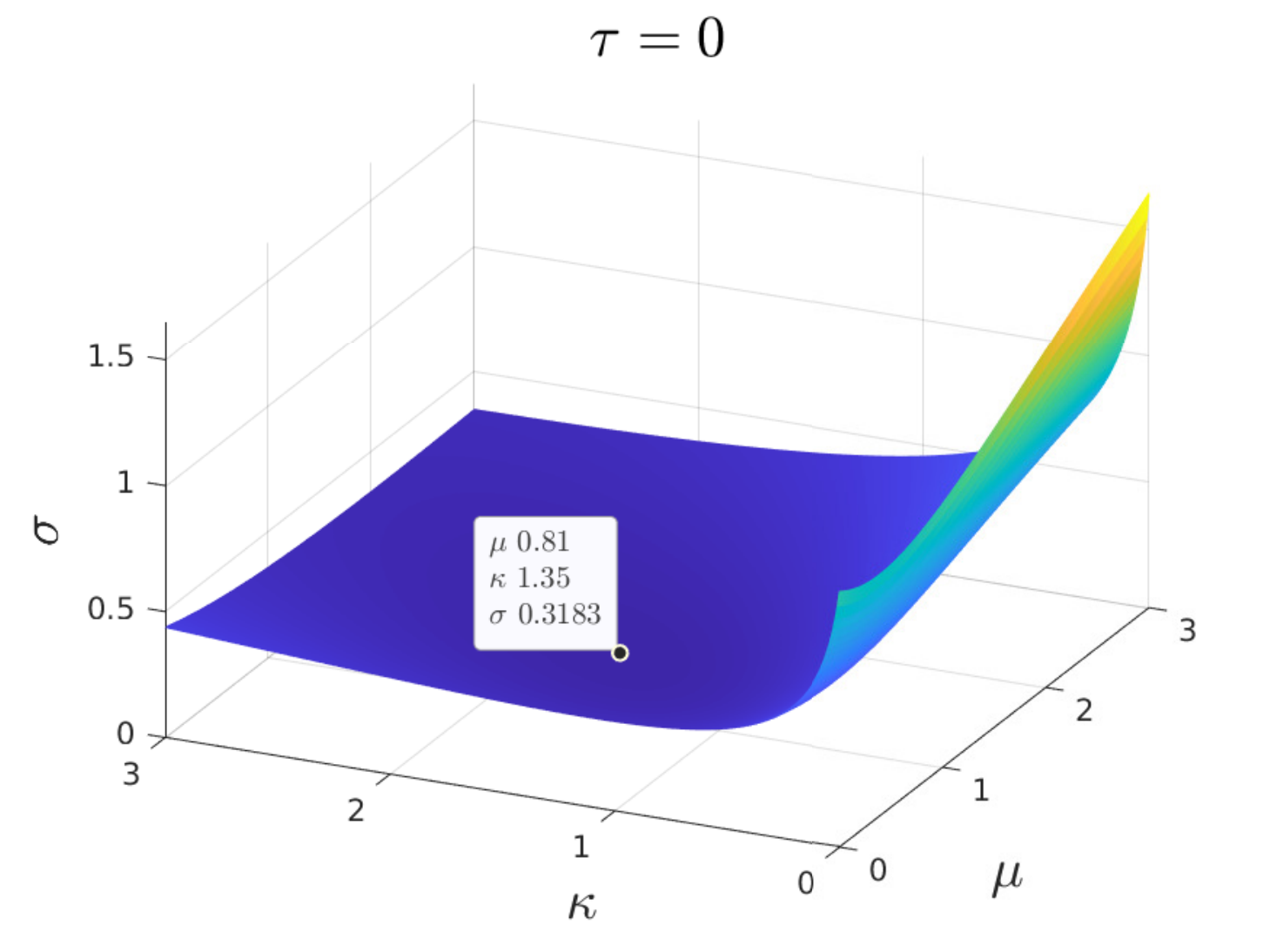}
        \caption{Joint $(\omega,\theta)$ control.}
        \label{fig: example1c}
    \end{subfigure}
    \hfill
    \begin{subfigure}[t]{.33\linewidth}
        \centering
        \includegraphics[width=\linewidth]{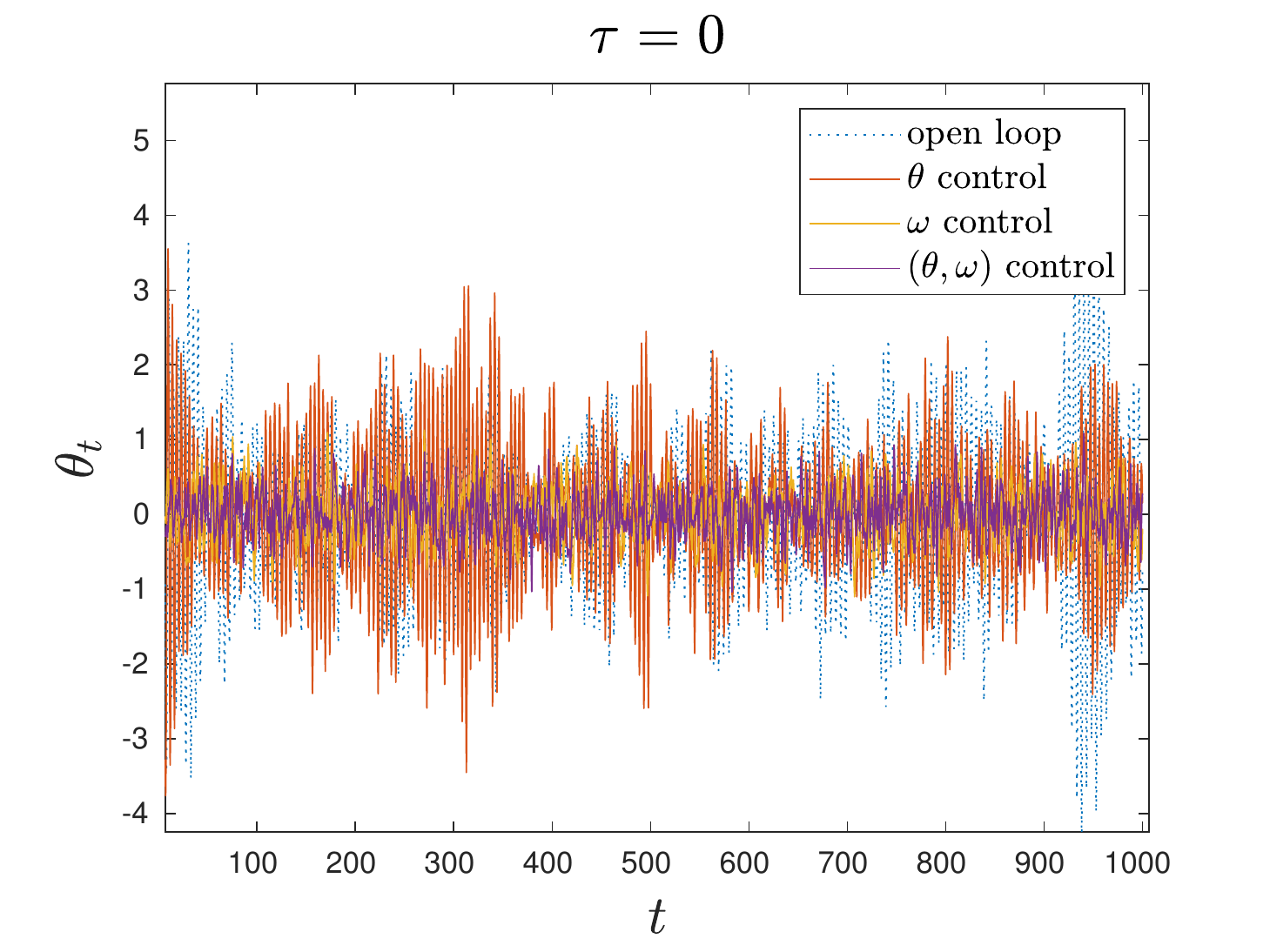}
        \caption{Solution realizations.}
        \label{fig: example1b}
    \end{subfigure}
    \medskip
    \begin{subfigure}[t]{.33\linewidth}
        \centering
        \includegraphics[width=\linewidth]{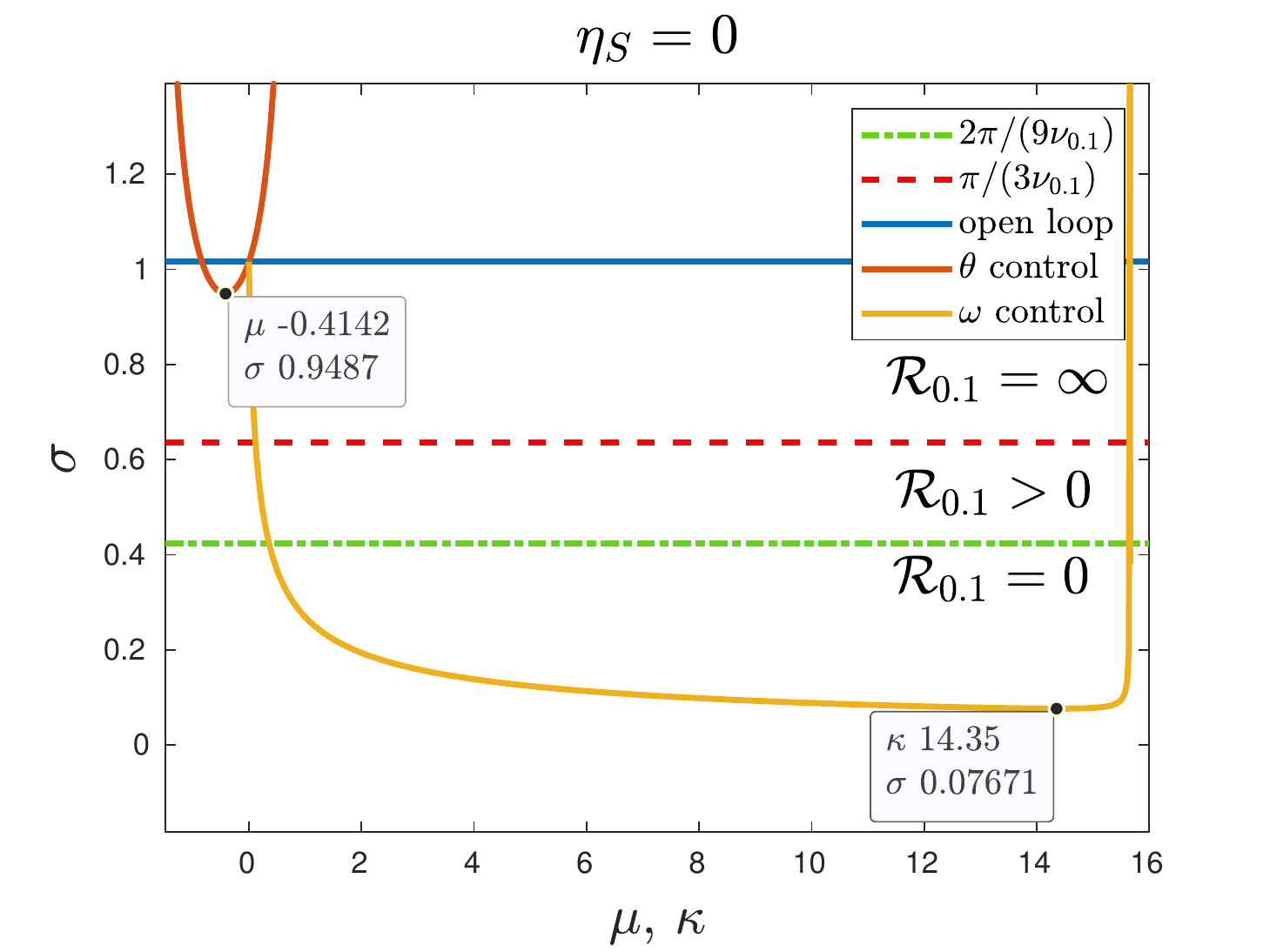}
        \caption{Separate $\omega$ and $\theta$ controls.}
        \label{fig: example1d}
    \end{subfigure}
    \hfill
    \begin{subfigure}[t]{.32\linewidth}
        \centering
	 \includegraphics[width=\linewidth]{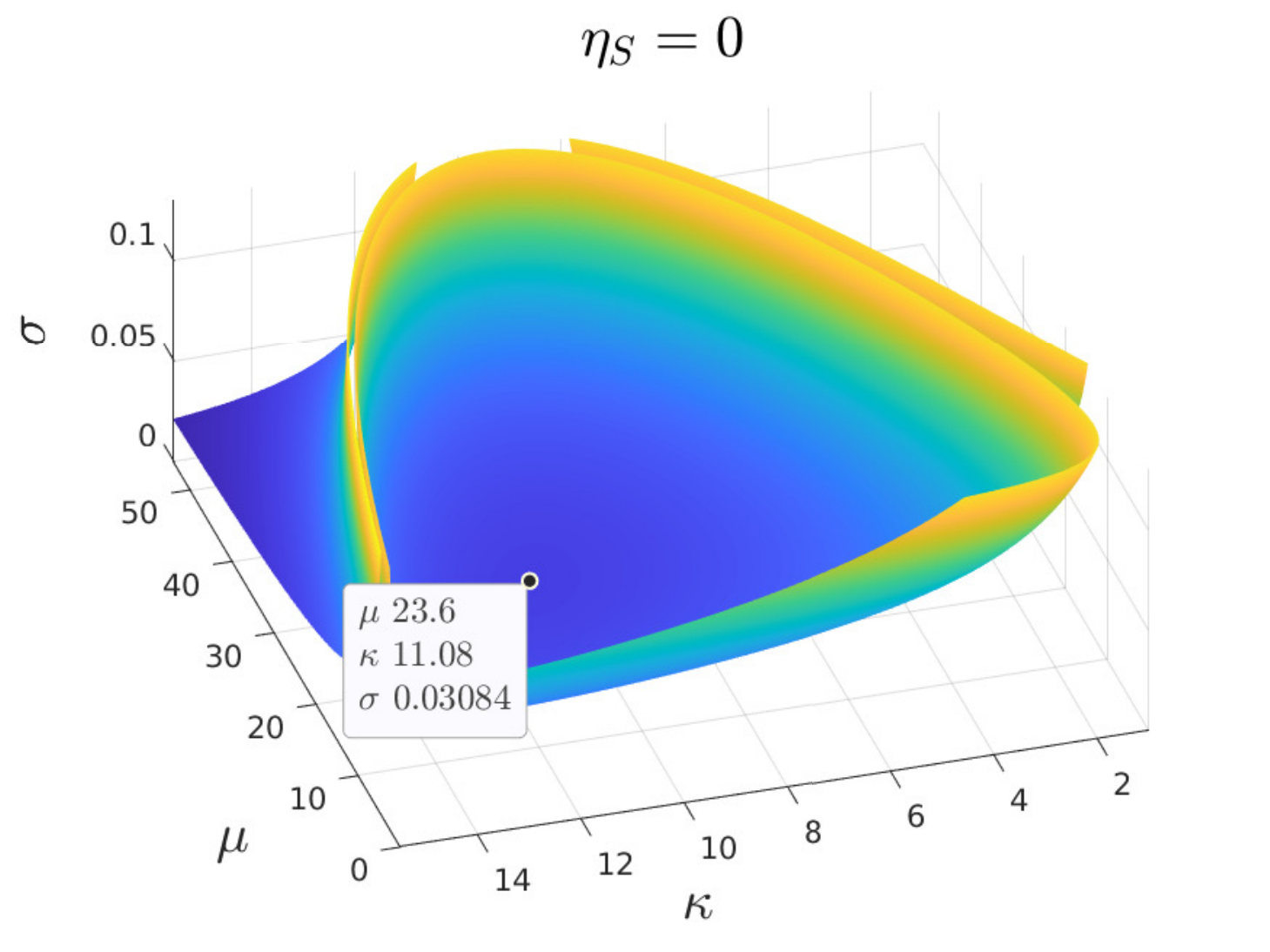}
	\caption{Joint $(\omega,\theta)$ control.}\label{fig: example1e}
    \end{subfigure}
    \hfill
    \begin{subfigure}[t]{.33\linewidth}
    \centering
	 \includegraphics[width=\linewidth]{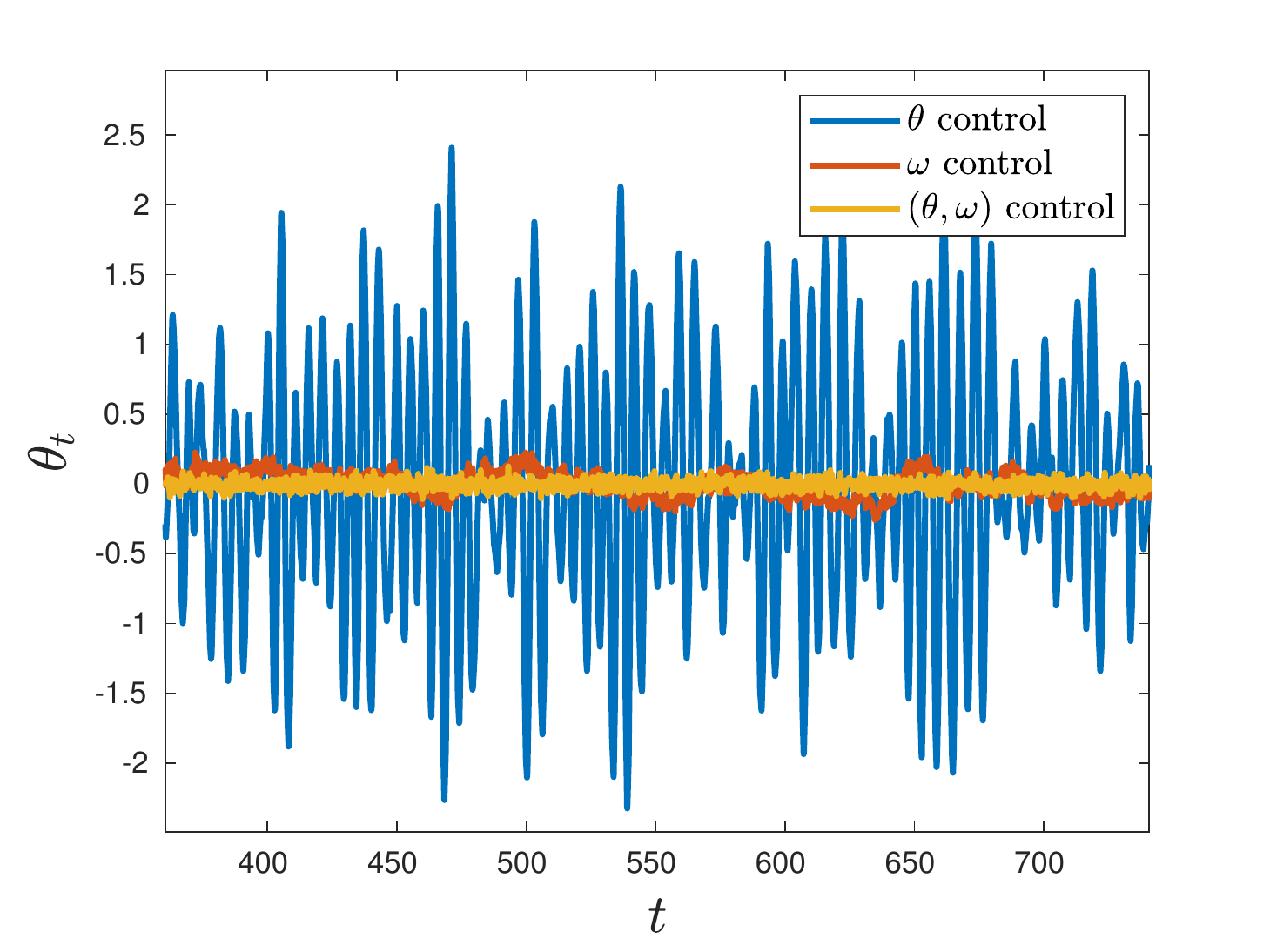}
	\caption{Solution realizations.}\label{fig: example1f}
    \end{subfigure}
	\caption{Risk-Aware Control for the Two-Machines System from Example 1. The first row depicts the case where adding measurement noise $(\eta'=0.3)$ to sensors, but the information is processed without time-delay. The second row illustrates the case where sensors have no measurement noise but lagged information processing with $\tau=0.1$. We compare single phase and frequency controls with joint phase and frequency, where the latter outperforms the former ones.}
	\label{fig: example1}
\end{figure*}

\subsection{The Two-Machine System}
Two synchronous generators are connected by a pure reactance $X=0.3$ per unit, as shown in Figure \ref{fig: 2machines}. The generators' transient reactances are $X_{d_1}=0.16 $ per unit and $X_{d_2}=0.20$ per unit respectively. Furthermore, the generators' inertia is set to $J= (2MJ)/(MVA)$ and the damping torque $\beta=0.15$. The system is operating in the steady-state with $E_1=1.2$ and $E_2=2$ per unit, at the equilibrium point $(0,0)\in \mathbb S$. In this case, the open-loop linearized swing equations read
\begin{equation*}
	\begin{split}
		2\, \ddot{\theta}^{(1)}_{t}=-0.15\,\dot{\theta}^{(1)}_t+1.584\,(\theta_{t}^{(2)}-\theta_{t}^{(1)})+\text{distrb}_1 \\
		2\, \ddot{\theta}^{(2)}_{t}=-0.15\,\dot{\theta}^{(2)}_t+1.584\,(\theta_{t}^{(1)}-\theta_{t}^{(2)})+\text{distrb}_2
	\end{split}
\end{equation*} 
\begin{figure}
	\begin{center}
		\begin{circuitikz}[american voltages]
			\draw
			(0,1) to [american voltage source, l_=$E_1\angle \theta_1$] (0,0); 
			\draw (0,1)  to [short, -] (0,2)
			  to [american inductor, l_=$X_{d_1}$] (2,2)
			  to [short, o-] (2,2)
			  to [american inductor, l_=$X$] (4,2)
			  to [short, -o] (4,2)
			  to [american inductor, l_=$X_{d_2}$] (6,2)
			  to [short, -] (6,2)
			  to [short] (6,1);
			  \draw (6,1) to [american voltage source, l_=$E_2\angle \theta_2$] (6,0);
			\draw (0,0) to (0,0) node[ground]{};
			\draw (6,0) to (6,0) node[ground]{};
		\end{circuitikz}
	\end{center}
	\caption{The two-machine system.}\label{fig: 2machines}
\end{figure}
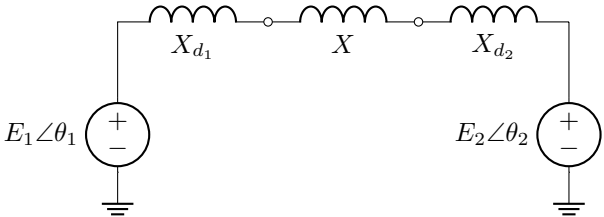 

The load volatility is modeled via white noise with standard deviation $\eta=0.7$. The structure is simple enough to consider the phase difference as $\theta_t=\theta_t^{(1)}-\theta_{t}^{(2)}$. When the uniform phase and the frequency feedback control gains are imposed on both generators, the phase difference satisfies the scalar stochastic delay differential equation
\begin{equation} \label{eq: 2machinesdifference}
	2\, \ddot{\theta}_{t}=-0.15\,\dot{\theta}_{t}-3.168\,\theta_t-\kappa\,\dot{\theta}_{t-\tau}-\mu\,\theta_{t-\mu}+ \text{distrb}
\end{equation} 
with gains $\mu>0$ and $\kappa>0$, that \eqref{eq: model} in this case reads
\begin{equation*}
	\begin{split}
	\mathcal A=\begin{bmatrix}
	0 && 1\\
	-1.584 && -0.075
	\end{bmatrix},~~~\mathcal K=\begin{bmatrix}
	0 && 0\\
	-\mu && -\kappa
	\end{bmatrix}\end{split}
\end{equation*} 

We will explore a few cases that fall within the scope of our framework. The distribution matrix is 
\begin{equation*}
\begin{split}
\mathcal H=\sqrt{2}\begin{bmatrix}
0 && 0 && 0 \\
0.35 && \eta' \,\mu && \eta'\, \kappa
\end{bmatrix}.
\end{split}
\end{equation*} 
where $\eta'$ is statistics of the sensor error. At first we find that the open-loop system \eqref{eq: 2machinesdifference} fluctuates converging to normal distributed with zero mean and standard deviation $\overline{\sigma}=1.0155$. Next, we consider a few scenarios as an application of Theorem \ref{thm: statistics}, Corollaries \ref{cor: risk2} and \ref{cor: risk3}. The results are all illustrated in Figure \ref{fig: example1}. 

\vspace{2mm}

\paragraph*{$\tau=0$} The first round of results assumes that WAC sensors have no time-delay but corrupted measurements as additive white noise with $\eta'=0.5$. We apply Corollary \ref{cor: risk2} to obtain the result. Firstly, we attempt to control the system using phase and frequency state feedback. The dependence of steady-state standard deviation concerning $\mu$ and $\kappa$ gains is illustrated in Figure \ref{fig: example1a}. We observe that the phase state feedback performs very poorly, and it only manages to decrease the standard deviation by $5.55\%$ (from 1.0155 to 0.9591). By increasing $\mu$ beyond 0.87, the closed-loop system performs worse than the open-loop system. On the other hand, the frequency state feedback decreases the initial standard deviation by $65.3\%$. Finally, the joint phase and frequency state feedback achieve a decrease by $68.66\%$ as illustrated in Figure \ref{fig: example1c}. Solution realizations are provided in Figure \ref{fig: example1b} for visual inspection of corresponding dynamics.
%
 
\vspace{2mm}
 
\paragraph*{$\eta'=0$} The second round of simulations regards perfect measurement recordings but lumped time-delay parameter as $\tau=0.1$. The narrative is similar to the previous case. Frequency state feedback control outperforms the phase state feedback control significantly, and the joint control outperforms both. Results are illustrated in Figures \ref{fig: example1d}, \ref{fig: example1e} and \ref{fig: example1f}. The maximum decrease from the open-loop system is $6.58\%$ for the phase state feedback control, $92.45\%$, and $96.96\%$ for the frequency state feedback control and the joint control.

\vspace{2mm}

\paragraph*{Risk-Aware Synthesis} 
Assume the systemic sets $U_\delta =\big( \frac{\pi}{3} \frac{1+\delta}{3/2+\delta},+\infty \big)$. The corresponding risk of phase incoherence is zero if the absolute value of phase difference is below $2\frac{\pi}{9}$, non-negative if it is up to $\frac{\pi}{3}$ and infinite, otherwise. The risk is calculated with the acceptance level $\varepsilon=0.1$, i.e., we evaluate the risk with the probability of staying outside the systemic set is at least $90\%$.
Figures \ref{fig: example1a} and \ref{fig: example1d} illustrate cut-off values of these systemic sets. Following the risk formulas in Corollary \ref{cor: risk2} or Corollary \ref{cor: risk3}, we can extract the feedback gain margins (for the individual control scenario) that ensure a zero, positive or infinite risk. We observe that, in general, the phase state feedback control fails to mitigate the risk of our systemic event at the imposed level of confidence. On the other hand, frequency and joint control perform considerably better. For example, in Figure \ref{fig: example1a} we see that frequency control achieves zero risk for gain in the interval $(0.41,2.76)$, and similarly for case of $\eta'=0$.


\subsection{The IEEE-39 standard}
Also known as the 10-machine New-England Power System, the IEEE-39 standard is a small-scale grid with ten generators and 39 buses. It is initially introduced in \cite{athay79} and its parameters are taken from \cite{pai1989}. The system's diagram is illustrated in Figure \ref{fig: IEEE39png}. The equivalent system created using the network reduction technique is illustrated in Figure \ref{fig: IEEE39reduced}, and it presents an exact reproduction of the transfer impedances of the reduced system as seen from its generator buses. Each non-diagonal element represents the admittance between each pair of generator buses. The bus admittance matrix of the IEEE-39 bus system is reduced to only include the generator nodes in the network. Therefore all of the off-diagonal elements represent the effective admittance between each generator node in the network. Thus, the width of every edge is proportional to the magnitude of admittance between connected generators.

\begin{figure}[t]
    \begin{subfigure}[t]{.45\linewidth}
        \centering
    	\includegraphics[width=\linewidth]{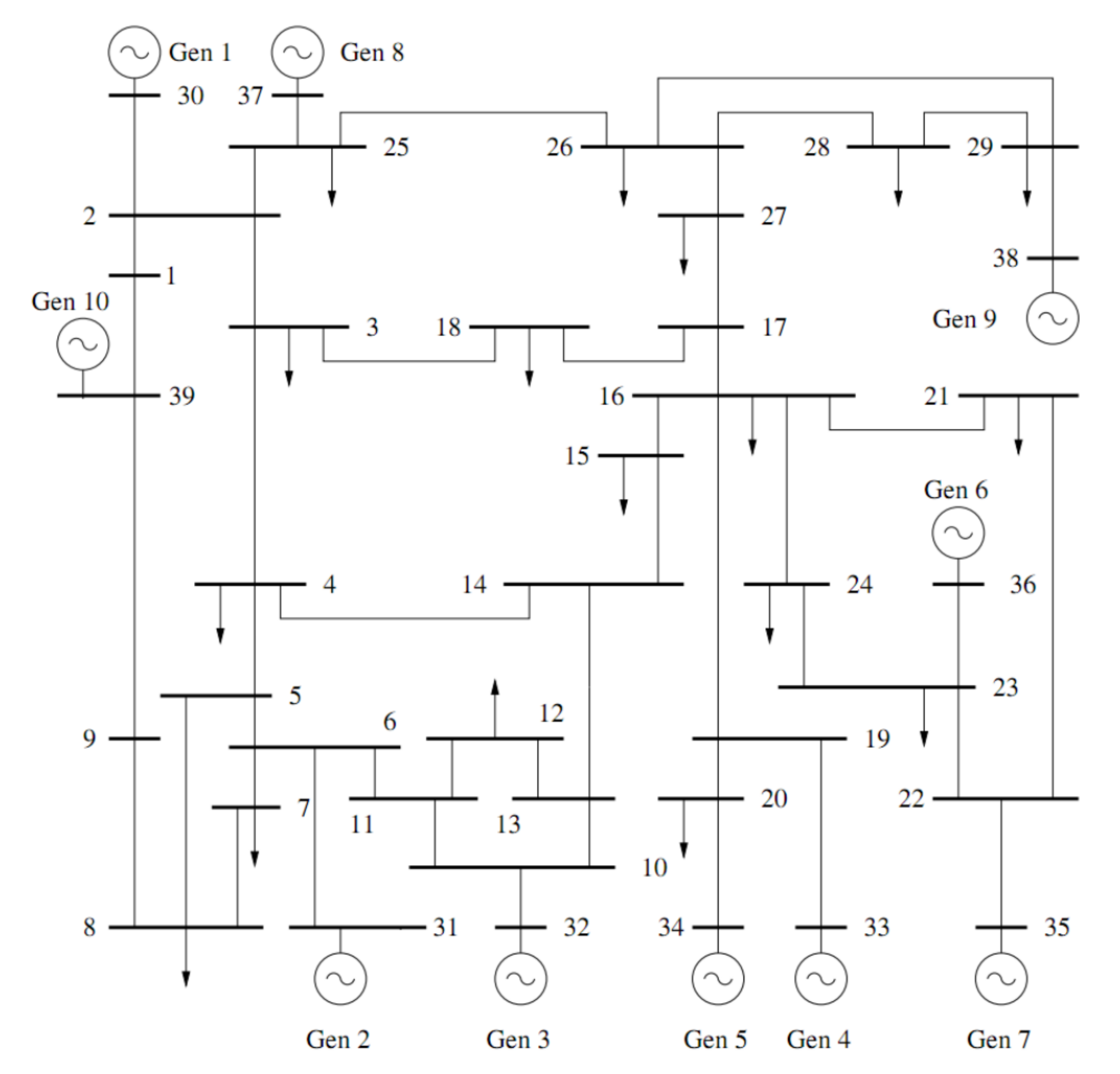}
    	\caption{Single Line Diagram.}
    	\label{fig: IEEE39png}
    \end{subfigure}
    \hfill
    \begin{subfigure}[t]{.54\linewidth}
        \centering
    	\includegraphics[width=\linewidth]{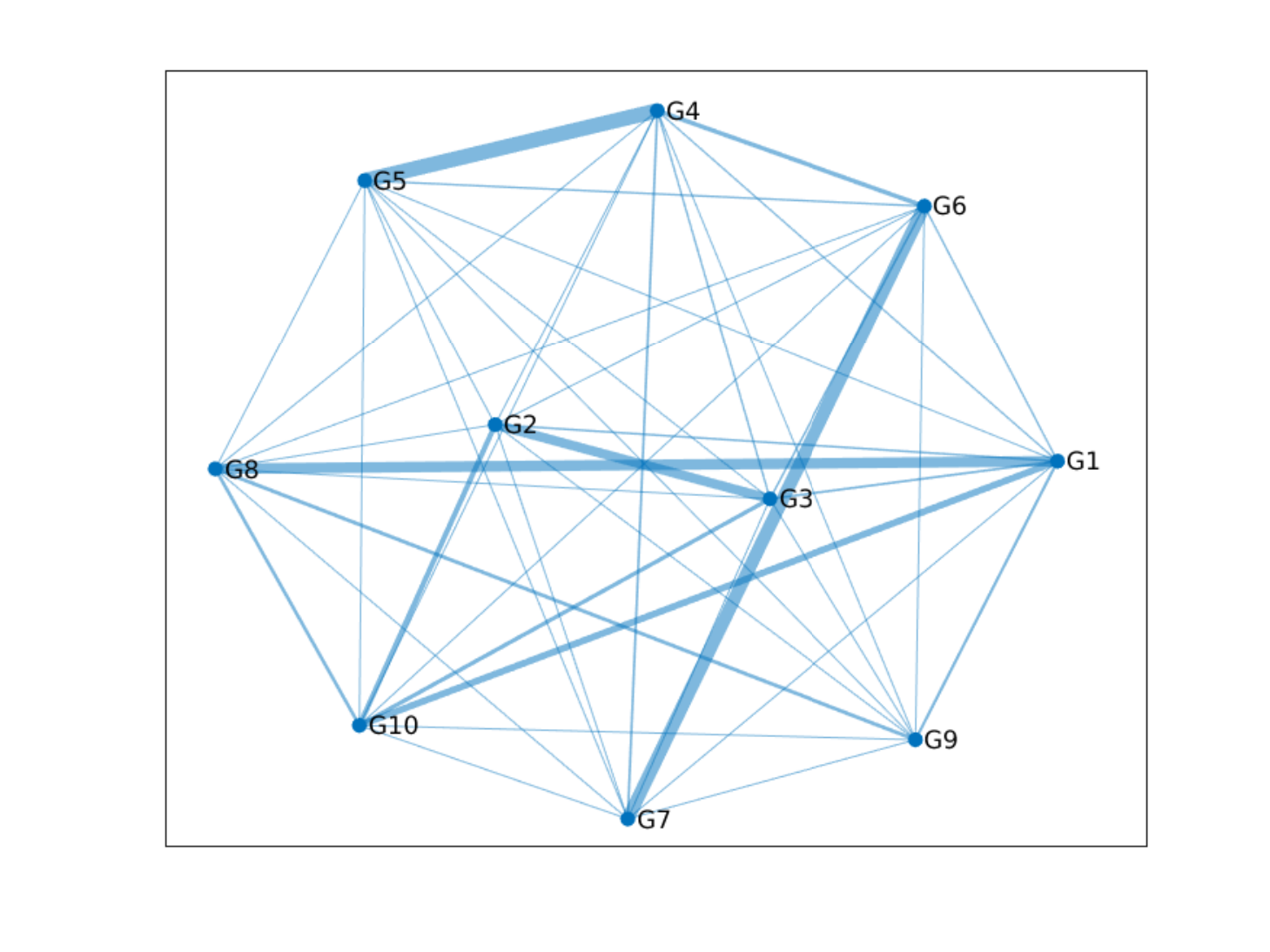}
    	\caption{Kron reduced graph.}
    	\label{fig: IEEE39reduced}
    \end{subfigure}
    \caption{The IEEE-39 Test-System for Example 2.}
    \label{fig: example2}
\end{figure}

We assume that the system is operating in steady-state equilibrium $(\theta^*,\omega^*)=(0,0)\in \mathbb S$ with some nominal voltages $\{E_{i}\}_{i=1}^{10}$ and the admittance matrix $Y$ is derived after reduction with system parameters from \cite{pai1989}. All generators have the same damping and inertia parameters as in Example 1, and load volatility has a standard deviation $\eta=1.1$. The sensor parameters are with $\eta'=0.2$ and $\tau=0.03$. 
We will examine the feasibility of synthesizing simultaneously diagonalizable controllers mitigating the risk of phase incoherence below deviation for various $\zeta$ limits of the systemic set with a probability of at least $95\%$. Following steps of Theorems \ref{thm: stability}, \ref{thm: statistics} and \ref{thm: risk} we calculate the optimal gains with nonzero eigenvalue of $L$, which is presented in Table \ref{tbl: example2}. For $Q$, the eigenvector matrix of $L$ (aligned with eigenvalues as in Table \ref{tbl: example2}), the phase and frequency optimal controllers are $M^*=Q \,\text{Diag}\{\mu_{j}^*\} \,Q^T$ and $K^*=Q\,\text{Diag}\{\kappa_{j}^*\}\, Q^T$, respectively.

Results are depicted in Figure \ref{fig: Example2}. The upper left plot illustrates the extreme risk values with and without optimal controllers $K^*$ and $M^*$ for various systemic set margins. We observe that as we decrease parameter $\zeta$ of the systemic set, the risk increases. In the open-loop case, all pairs of generators fall into incoherence concerning $\mathcal R_{0.05}$ for values of roughly $\zeta=0.7$ and below. For the case of optimized control, we do not see the divergence of any pair of generators for values of $\zeta$ larger than $0.58$, while even for ultimately conservative strict systemic sets, some pairs are still risk-free. In the remaining three plots, we illustrate the systemic risk distribution over all pairs of generators\footnote{The enumeration of pairs is held row-wise: The first value accounts for $(1,2)$ pair, the second for $(1,3)$, the tenth for $(2,3)$, and so forth.}. We observe that the optimal controllers can mitigate the risk of phase incoherence efficiently. We also note that the minimal risk controllers are following Theorem \ref{thm: limit} when it comes to relative values of minimal standard deviation and systemic set. The absence of open-loop risk values for the first two plots indicates infinite risk in all of its pairs of generators' phases. On the contrary, the absence of optimal risk values in the bottom right plot indicates the optimal controllers achieve risk-free fluctuations for the systemic set with $\zeta=\frac{\pi}{4}$, also indicated in the first subplot of Figure \ref{fig: Example2}.


\begin{figure}[t]
    \begin{subfigure}[t]{.49\linewidth}
        \centering
    	\includegraphics[width=\linewidth]{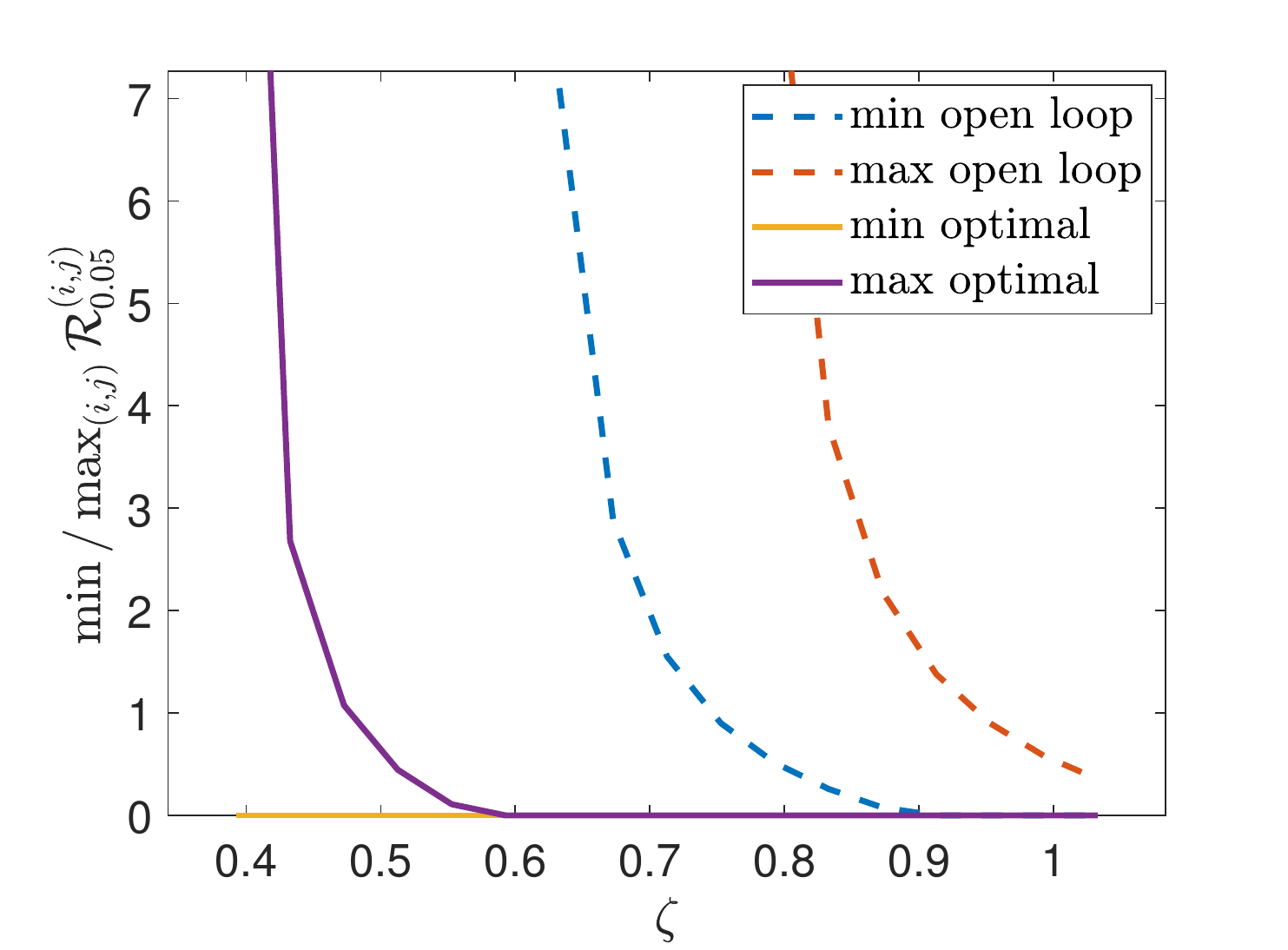}
    \end{subfigure}
    \hfill
    \begin{subfigure}[t]{.49\linewidth}
        \centering
    	\includegraphics[width=\linewidth]{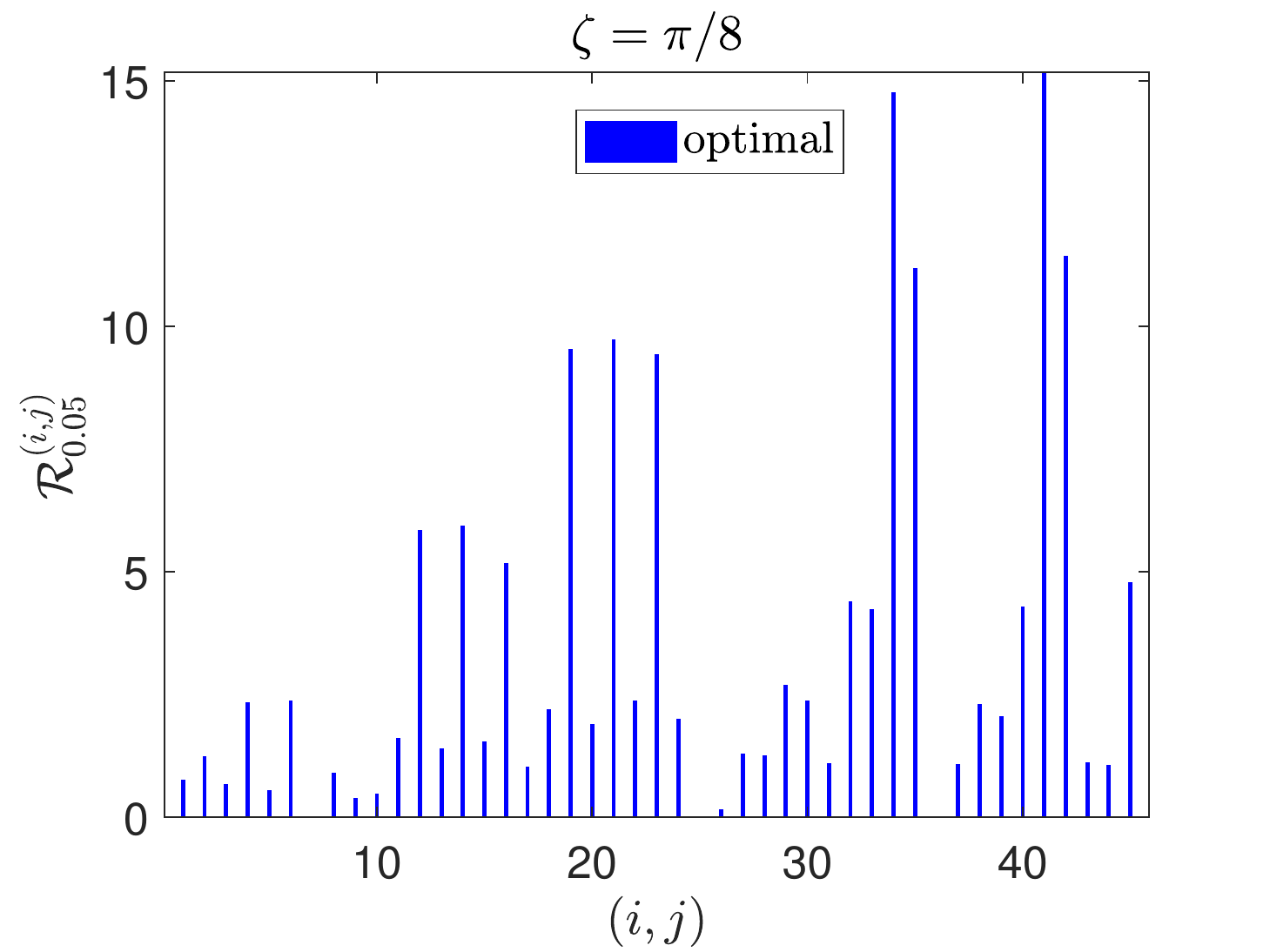}
    \end{subfigure}
    \medskip
    \begin{subfigure}[t]{.49\linewidth}
        \centering
    	\includegraphics[width=\linewidth]{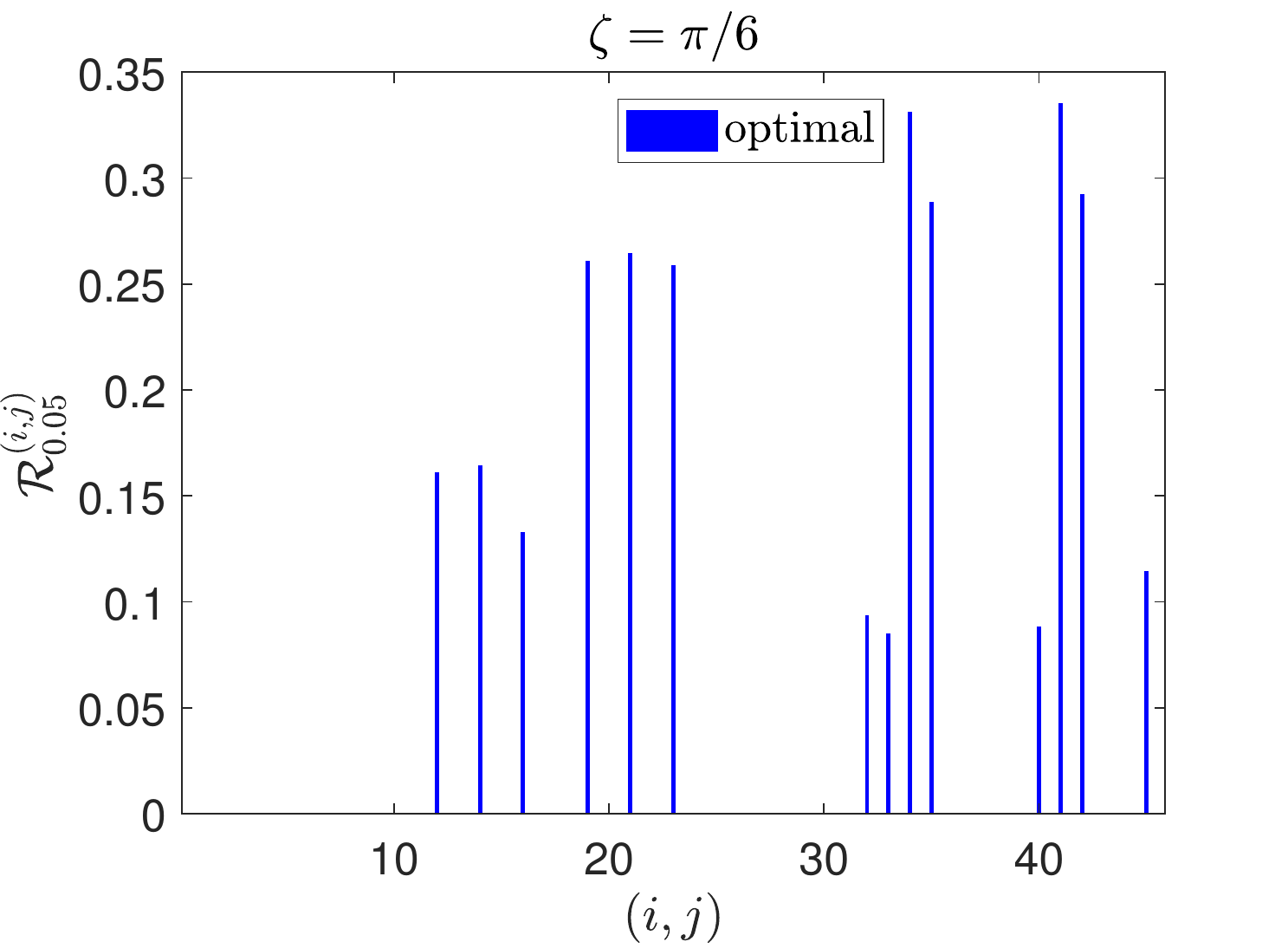}
    \end{subfigure}
    \hfill
    \begin{subfigure}[t]{.49\linewidth}
        \centering
    	\includegraphics[width=\linewidth]{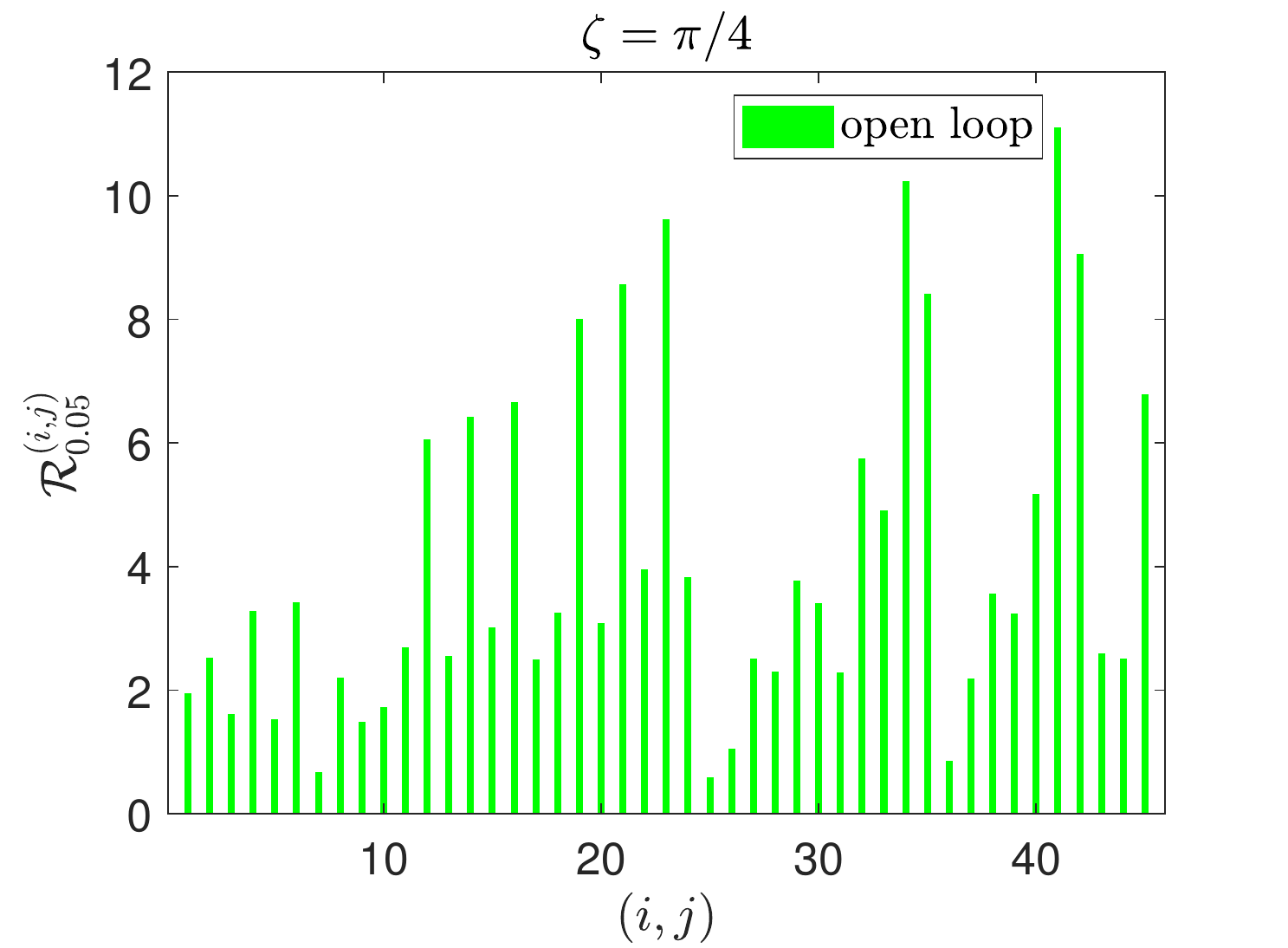}
    \end{subfigure}
    \caption{Risk measures on Example 2. (Upper Left): Sub-plot of maximum and minimum risk concerning systemic set parameters. (Upper Right, Lower Left, Lower Right): Sub-plots of steady state risk distribution over all 45 pairs of generators for the optimal at $\zeta=\pi/8$, $\zeta=\pi/6$ and for the open-loop case at $\zeta=\pi/4$, respectively.}
	\label{fig: Example2}
\end{figure}

\begin{table}[t]                           
	\centering	
	\resizebox{\linewidth}{!}
	{ 
		\begin{tabular}{c c c c}
            \toprule
            $\hspace{0.7cm} \lambda_j \hspace{0.7cm}$ & $\hspace{0.7cm} \mu^{*}_j \hspace{0.7cm}$ & $\hspace{0.7cm} \kappa_j^{*} \hspace{0.7cm}$ & $\hspace{0.7cm} \mathfrak f \hspace{0.7cm}$ \\
            \midrule
            23.8762 & 0.25   & 2.75 &  0.0672 \\ [0.5ex]
            31.8500 & 0.20 & 2.75 & 0.0584 \\[0.5ex]
            34.9876 & 0.15 & 2.75 & 0.0557 \\[0.5ex]
            44.5137  & 0.10  & 2.75 & 0.0495 \\[0.5ex]
            55.6556 &  0.10  & 2.70  &  0.0444 \\ [0.5ex]
            64.0023 & 0.05 & 2.70  &   0.0415 \\[0.5ex]
            88.7335  &  0.05  & 2.70 & 0.0355 \\[0.5ex]
            94.8997 & 0.05  & 2.70  & 0.0343 \\[0.5ex]
            103.9912  & 0.05 & 2.70  &  0.0329 \\
			\bottomrule
	    \end{tabular} }                        
    \caption{Example 2. $(\mu*,\kappa*)$ minimizers of $\mathfrak{f}$ in Theorem \ref{thm: statistics}.}
    \label{tbl: example2}
\end{table}


\section{Discussion}    \label{sect: conclusion}
The notion of transient stability in power networks is traditionally studied via generalized Lyapunov methods that compute the critical energy the system can absorb before the post-fault system is restored into some stable state. However, these methods provide conservative results and do not scale with  network size. To this end, we adopt a risk-based framework to study robustness of a class of wide-area control laws for phase incoherence in presence of exogenous noise and communication time-delay. Our results suggest that time-delay and erroneous measurements can severely impact and increase the systemic risk of phase incoherence, which potentially can steer the network to undesirable modes of operation. Our analysis relies on stability results and the evaluation of implicit spectral functions. We fully characterize the stability margin and establish delay-and-noise induced fundamental limits of a class of WAC laws that show there exists a trade-off between the best achievable levels of risk of phase incoherence and performance. 


The main body of the computational load in our setup was to calculate eigen-decomposition of the underlying Laplacian matrix. Thus, we do not need to calculate the spectral functions for every individual mode, instead we have to evaluate the functions over the Laplacian eigen-spectrum. This dramatically simplifies risk analysis and reduces the time complexity to that of  decomposing a $n \times n$ symmetric matrix, which is $\mathcal O(n^3)$.

Assumptions \ref{assum: commute} and \ref{assum: delay} have facilitated our analysis to obtain several analytical formulae. Both lumped time-delay $\tau>0$ and simultaneous diagonalization helped us fully classify the stability region of the resulting  closed-loop system as well as deriving closed-form expressions for the risk measure, although with excessively complicated formulae. One can relax either or both conditions and explore generalized setups within the context of systemic risk by following the same analysis steps outlined in this work. The price to pay is that it will be impossible to obtain  closed-form and the exact representation of stability regions and risk expressions. General feedback matrices and heterogeneous  time-delays typically are handled with standard textbook techniques  \cite{gu2003stability} that provide sufficient and conservative conditions. Then, the risk measure can probably be approximated with an upper bound. Despite this drawback, generalized universal limitations and trade-offs can still be derived along the lines described in the paper.

\appendix \label{appendix}

\subsection{Proofs}\label{ap: proofs}
\subsubsection{Theorem \ref{thm: stability}}
Given the structure of $\mathbb S$ regarding the phase of generators, we explore conditions for convergence of the zero-input system \eqref{eq: model} with respect to $(1_n \, \rho,0)$ for some $\rho\in \mathbb R$. Based on Assumptions \ref{assum: uniformity} - \ref{assum: commute}, it is convenient to work with the transformed coordinates 
$$
\vartheta_t : = Q^T \theta_t ~~~\text{and}~~~ \varpi_t : =Q^T \omega_t.
$$ 
The desirable type of convergence is equivalent to $\vartheta^{(j)}_t\rightarrow 0$ for $j>1$ and $\vartheta^{(1)}_t\rightarrow \rho$, while $\varpi_t^{(j)} \rightarrow 0$ for all $j\geq 1$. The transformed dynamics
\begin{equation*}
\begin{split}
\begin{bmatrix}
\dot{\vartheta}_t \\ \dot{\varpi}_t
\end{bmatrix}&=\begin{bmatrix}
O_n & I_n \\
-\Lambda_L & -\Lambda_D
\end{bmatrix} \begin{bmatrix}
\vartheta_t \\
\varpi_t
\end{bmatrix}+\begin{bmatrix}
O_n & O_n \\
-\Lambda_{M} & -\Lambda_{K}
\end{bmatrix} \begin{bmatrix}
\vartheta_{t-\tau} \\
\varpi_{t-\tau}
\end{bmatrix}
\end{split}
\end{equation*} are fully decomposed and facilitate separated analysis for $j=1$ and $j>1$. By linearity, we seek solutions of the form $e^{\eta t}$ with $\eta=\eta_\tau \in \mathbb C$. Condition 
$\mathrm{Re}\{\eta\}=0$ signifies the onset of marginal stability and, in most cases, instability of the zero solution. Locating such solutions is accomplished with the study of the characteristic equation $c_j(\eta)=0$ of the  $j^{th}$ pair $\big(\vartheta_t^{(j)},\varpi_t^{(j)}\big)$, that reads 
\begin{equation}\label{eq: char} 
	c_j(\eta)=\eta^2 + d \eta + \kappa_j \, \eta \, e^{-\eta \tau } + \lambda_j+\mu_j \, e^{-\eta \tau}.
\end{equation} 

Analytical results on solutions of $c_j(\eta)=0$ for $\tau \geq 0$ can be found in standard textbooks of the subject, such as \cite{kuang1993delay}, from where we draw results to establish our proof. Heuristically speaking, $c_j(\eta)=0$ may have zero, one, or two distinct imaginary roots. The stability of the time-delayed system dynamics is then associated with the stability properties for $\tau=0$. In the case of no imaginary roots for $c_j(\eta)=0$, the stability properties of decomposed sub-system are retained for all time-delays, e.g., if it is unstable at $\tau=0$, the system will remain unstable for all $\tau \geq 0$. When $c_j(\eta)$ attains precisely one imaginary roots, the instability at $\tau=0$ will be retained for all $\tau \geq 0$. Stability at $\tau=0$ will switch only once at a critical time-delay, beyond which the system is unstable. For the two distinct imaginary roots of $c_j(\eta)=0$, there will be a finite number of switches between instability and stability before the termination into the final instability .
For more details, we refer to Theorem 3.1 in Sec 3.3  of  \cite{kuang1993delay}. Since the requirements for $j=1$ differ from $j>1$, we distinguish between the two cases. Finally, all results will be expressed after suppressing time-delay $\tau>0$ into the rest parameters for convenience
\begin{equation}\label{eq: transform} 
	d \leftrightarrow d\tau,~~\lambda_j \leftrightarrow \lambda_j\tau^2,~~\mu_j \leftrightarrow \mu_j \tau^2,~~\kappa_j \leftrightarrow \kappa_j \tau.
\end{equation}

\vspace{2mm}

\noindent{\underline{Case $j=1$:}} Here we have $\lambda_1=0$. For $\mu_1=0$ as well, the sub-system $j=1$ becomes decoupled in dynamics of $\varpi^{(1)}_t$:
$$\frac{d}{dt}\vartheta_t^{(1)}=\varpi^{(1)}_t~~~~~\text{and}~~~~~\frac{d}{dt}\varpi_t^{(1)}=-d\varpi^{(1)}_t-\kappa_1 \varpi_{t-\tau}^{(1)} $$
The convergence of $\vartheta_t^{(1)}$ in $\mathbb R$ is associated to the asymptotic stability of $\varpi^{(1)}_t$ to zero. This is achieved if and only if $-d<\kappa_1<d$, or $|\kappa_1|>d$ with $\tau<\overline{\tau}_{s;k}$, for critical $\overline{\tau}_{s;k}$ as in \eqref{eq: bound1}. In addition, for $\kappa_1=0$, we can completely solve the system and calculate the converging point $\rho$ as in the statement of Theorem \ref{thm: stability}. Inducing transformation \eqref{eq: transform}, the stability region is equivalent to the set $\mathbb W_0(s;k)$, evaluated at $(d\tau, 0~;~0 , \kappa_1 \tau)$.

When $\mu_1\neq 0$, $c_1(\eta)=0$ can have only one imaginary root (see the proof of Theorem 3.1 in \cite{kuang1993delay}). The stability at $\tau=0$ implies the stability for the range of $\tau\in \Big[0,\tau_{s;k}^{0,+}\Big)$, where $\tau_{s;k}^{0,+}$ as in \eqref{eq: bound2}. After time-delay suppression with the new coordinates \eqref{eq: transform}, the range of parameters is covered by $\mathbb W_2(s;k)$, which is evaluated at $(d\tau, 0 ~;~ \mu_1\tau^2, \kappa_1 \tau)$.

\vspace{2mm}

\noindent{\underline{Case $j>1$:}} For projections of dynamics over the disagreement space, we strictly require the asymptotic stability. The first step is to verify this for $\tau=0$. Applying the Hurwitz criterion, the stability holds if and only if \cite{kuang1993delay}
\begin{equation}\label{eq: stabilitytau=0}
	\kappa_j+d>0 ~~\text{and}~~\lambda_j+\mu_j>0.
\end{equation}
\noindent $I.$ The equation $c_j(\eta)=0$ attains no imaginary roots if and only if $\lambda_j^2>\kappa_j^2$ and $\kappa_j^2+2\lambda_j-d^2\leq 2\sqrt{\lambda_j^2-\mu_j^2}$. The stability is equivalent to the transformed coordinates \eqref{eq: transform}, which belong to the set $\mathbb W_1(s;k)$. 

\noindent $II.$  The equation $c_j(\eta)=0$ attains exactly one imaginary root if and only if $\lambda_j^2\leq \mu_j^2$. The asymptotic stability is guaranteed for a finite range of time-delay $\tau\geq 0$ if and only if the condition \eqref{eq: stabilitytau=0} holds and the maximum allowed time-delay is $\tau^{0,+}_{s;k}$ as  defined in \eqref{eq: bound2}. Implementing \eqref{eq: transform}, the parameter region that fulfills this condition belongs to $\mathbb W_2(s;k)$.
 
 

\noindent $III.$ The equation $c_j(\eta) = 0$ attains two distinct roots with the initial stability at $\tau=0$. As time-delay is increasing, there will be a finite number of stability switches before the eventual instability. The stability region for time-delay $\tau$ is determined with set $\mathfrak T_{k;d}$ as in \eqref{eq: tauset}. The rest of parameter values ought to be members of $\mathbb W_2(s;k)$, see Theorem 3.1 in  \cite{kuang1993delay}. \hfill$\blacksquare$

\begin{table*}\centering
    \begin{minipage}{1.0\textwidth} 
    \hrule
    \begin{align*}
    \mathbb W_0(s;k)&=\bigg\{s\in \mathbb R_+^2,~k\in \mathbb R^2~:~s_2=k_1=0,~\big\{ ~ |k_2|<s_1 ~ \} \cup \big\{~  k_2 >s_1,~ \sqrt{k_2^2-s_1^2}<\text{arccot}\big(-s_1 / \sqrt{k_2^2-s_1^2} \big)~ \big\} \bigg\} \\
    \mathbb W_1(s;k)&=\bigg\{s\in \mathbb R_+^2,~k\in \mathbb R^2~:~s_2^2>k_1^2,~k_2+s_1>0,~k_1+s_2>0,~k_2^2+2s_2-s_1^2\leq 2\sqrt{s_2^2-k_1^2} \bigg\} \\
    \mathbb W_2(s;k)&=\bigg\{s\in \mathbb R_+^2,~k\in \mathbb R^2~:~s_2^2\leq k_1^2,~ k_2+s_1>0,~ k_1+s_2>0,~\gamma_+(s;k) < \varphi_+(s;k) \bigg\} \\
    \mathbb W_3(s;k)&=\bigg\{s\in \mathbb R_+^2,~k\in \mathbb R^2~:~s_2^2> k_1^2,~ k_2+s_1>0,~ k_1+s_2>0,~k_2^2+2s_2-s_1^2> 2\sqrt{s_2^2-k_1^2},~ \big(\gamma_{\pm}(s;k),\varphi_{\pm}(s;k)\big)\in \mathfrak{I}_{s;k}\bigg\}
    \end{align*}  
    \hrule 
    \end{minipage}
    \centering 
    \caption{Parameter Sets of Stability Area}
    \label{table: stability}
\end{table*}

\subsubsection{Theorem \ref{thm: statistics}}
We write 
$$
y_t = C \begin{bmatrix}
\theta_t \\
\omega_t
\end{bmatrix}= \begin{bmatrix}
B_n & O_{r\times n}
\end{bmatrix} \begin{bmatrix}
\theta_t \\
\omega_t
\end{bmatrix},
$$ for $r=\frac{n(n-1)}{2}$. Vector $(\theta_t,\omega_t)^T$ is normally distributed and so will $y_t$, although with different first and second moments. In particular, the mean vector of $y_t$ will be $C\,T_t$ and the covariance matrix will be $$  \Sigma_t =  \int_{0}^{t} C \,\Phi(s) \, \mathcal H \mathcal H^T \, \Phi^T(s)\, C^T \,ds $$ 
%
The transition matrix $\Phi(t)$ of system \eqref{eq: model} can be written in terms of decomposed transition matrix $\tilde{\Phi}(t)$ as follows $$\Phi(t)=\begin{bmatrix}
Q & O_n \\ O_n & Q
\end{bmatrix} \tilde{\Phi}(t) = \begin{bmatrix}
Q & O_n \\ O_n & Q
\end{bmatrix} \begin{bmatrix}
\tilde{\Phi}_{\vartheta\vartheta}(t) & \tilde{\Phi}_{\vartheta\varpi}(t) \\ \tilde{\Phi}_{\varpi\vartheta}(t) & \tilde{\Phi}_{\varpi\varpi}(t)
\end{bmatrix}. $$ Here the $(i,j)$-th element of matrix block $\tilde{\Phi}_{\gamma \zeta}$ represents the response of the $i$-th element of $\gamma$ state sub-vector, to excitation implemented to the $j$-th element of $\zeta$ state sub-vector. It is not hard to check that $C\,T_t$ converges to zero. Indeed, $$C\,T_t = \int_{-\tau}^0 C \Phi(t-s)\begin{bmatrix}
\phi^{\theta}(s) \\ \phi^{\omega}(s)
\end{bmatrix}\,d\mu(s)~~~~\text{and}$$\begin{equation*} \begin{split} C \Phi(s)&= \begin{bmatrix}
B_n Q & O_{r\times n}
\end{bmatrix} \begin{bmatrix}
\tilde{\Phi}_{\vartheta\vartheta}(s) & \tilde{\Phi}_{\vartheta\varpi}(s) \\ \tilde{\Phi}_{\varpi\vartheta}(s) & \tilde{\Phi}_{\varpi\varpi}(s)
\end{bmatrix} \begin{bmatrix}
Q^T & O_n \\ O_n & Q^T
\end{bmatrix}   \\ & = \begin{bmatrix}
B_n\, Q \, \tilde{\Phi}_{\vartheta\vartheta}(s)\,Q^T &  B_n \, Q\, \tilde{\Phi}_{\vartheta\varpi}(s)\,Q^T
\end{bmatrix}\end{split} 
\end{equation*} 
Now, since $B_n 1_n=0_{r}$, $B_n Q$ is an $r\times n$ matrix with zero first column. This will, in turn, nullify the first element of diagonal sub-matrices $\tilde{\Phi}_{\vartheta\vartheta}$ and $\tilde{\Phi}_{\vartheta\varpi}$, which refer to possibly marginally stable projections of $\theta_t$ and $\omega_t$. The remaining dynamics are, in the view of Theorem \ref{thm: stability} exponentially stable, validating the claim on $C\, T_t$. Consequently, the non-constant elements of $\Sigma_t$ are also exponentially stable, yielding integrable integrands and well-posedness of covariance matrix as $t\rightarrow +\infty$. Now, $$\mathcal H \mathcal H^T=\begin{bmatrix}
O_n & O_n \\ O_n & H H^T
\end{bmatrix}=\begin{bmatrix}
O_n & O_n \\ O_n & \frac{\eta^2}{J^{2}}I_n+\eta'^2(MM^T+KK^T)
\end{bmatrix}$$ the integrand of $\Sigma_t$ equals 
\begin{equation*}
	\begin{split} 
		&B_n\, Q \, \tilde{\Phi}_{\vartheta \varpi}(s) \, Q^T\,  H  H^T\, Q \, \tilde{\Phi}_{\vartheta\varpi}(s)\, Q^T B_n^T  = \\ 
		&=\text{diag}\bigg\{\frac{\eta^2_E}{J^2}+\eta'^2\big(\mu_j^2+\kappa_j^2\big)\bigg\} B_n\, Q \, \tilde{\Phi}_{\vartheta \varpi}^2(s)\, Q^T B_n^T.   
	\end{split} 
\end{equation*} 
Using the similar argumentation, the first element of diagonal matrix has no contribution as it is nullified by $B_n Q$. Eventually, $\lim_t \Sigma_t$ is equal to 
$$
\text{diag}\bigg\{\frac{\eta^2_E}{J^2}+\eta'^2\big(\mu_j^2+\kappa_j^2\big)\bigg\} \int_{0}^{\infty}  B_n \, Q\, \tilde{\Phi}^2_{\vartheta\varpi}(s)\,Q^T\, B_n^T \,ds. 
$$  
Now, let us note that $\tilde{\Phi}_{\vartheta\varpi}(t)=\text{diag}\{\tilde{\Phi}_{\vartheta\varpi}^{(jj)}(t)\}$ is a transition matrix of the decomposed dynamics. In addition to being a diagonal matrix, all but the first non-zero elements are exponentially stable. From Parseval's theorem
\begin{equation*}
\int_0^\infty\big[\tilde{\Phi}_{\vartheta\varpi}^{(jj)}(t)\big]^2\,dt=\frac{1}{2\pi}\int_{\mathbb R}\frac{d\gamma}{\|c_j(i\gamma)\|^2},
\end{equation*} 
for $c_j(\eta)$ is the characteristic equation of the $j$-th decomposed sub-system from \eqref{eq: char} and $i^2=-1$. Tedious algebra with the transformation \eqref{eq: transform} yields $$\int_0^\infty\big[\tilde{\Phi}_{\vartheta\varpi}^{(jj)}(t)\big]^2\,dt = \frac{\tau^{3}}{2\pi } f\big(d\tau, \lambda_j\tau^2 ; \mu_j\tau^2, \kappa_j \tau \big),$$ for $j=2,\dots,n$  and function $f(s;k)$ is defined in \eqref{eq: f}. Then the result follows with simple algebra. \hfill$\blacksquare$

\subsubsection{Theorem \ref{thm: risk}}
Let us denote by $\overline{y}^{(i,j)}$ the element of $\lim_t y_t$ that describes the steady-state phase difference between the $i^{th}$ and $j^{th}$ generators. By virtue of Theorem \ref{thm: statistics}, $\overline{y}^{(i,j)}$
is normally distributed with zero mean and variance 
\begin{align*}
	\sigma_{ij}^2 & =\frac{1}{2\pi}  \sum_{k=2}^n (e_i-e_j)^TB_n Q \,\text{diag}\{\mathfrak f_k\}\, Q^T B_n^T (e_i-e_j) \\
	& = \frac{1}{2\pi}  \sum_{k=2}^n \big(q_{ik}-q_{jk} \big)^2 \,\mathfrak f_k.
\end{align*}
We calculate the risk of phase incoherence at $\overline{y}^{(i,j)}$ concerning $\{U_\delta\}_{\delta >0}$ as follows: Given $\varepsilon\in (0,1)$ and $\delta>0$, we write 
\begin{equation*}
	\mathbb P\big(~\big|\overline{y}^{(i,j)}\big|\in U_\delta~\big)<\varepsilon
	~\Leftrightarrow~ 
	\mathbb P \left( ~\zeta~ \frac{1+\delta}{c+\delta} < |\,\overline{y}^{(i,j)}|~ \right)<\varepsilon,
\end{equation*} 
which is equivalent to 
$$ 
\mathbb P\left(-\zeta~\frac{1+\delta}{c+\delta}<\overline{y}^{(i,j)}< \zeta~\frac{1+\delta}{c+\delta}~ \right) \geq 1-\varepsilon. 
$$ 
Clearly, $\mathcal R^{ij}=\mathcal R(\overline{y}^{ij})$ is equal to $\delta>0$ where $$\inf\bigg\{\delta>0 ~:~ \frac{1}{\sqrt{2\pi}\sigma_{ij}}\int_{-\zeta\frac{1+\delta}{c+\delta}}^{\zeta\frac{1+\delta}{c+\delta}}e^{-\frac{t^2}{2\sigma_{ij}^2}}\,dt>1-\varepsilon \bigg\}, $$ equivalently, 
\begin{equation}\label{eq: riskcond}
	\inf\bigg\{\delta>0 ~:~ \frac{1}{\sqrt{2\pi}}\int_{-\frac{\zeta}{\sigma_{ij}}\frac{1+\delta}{c+\delta}}^{\frac{\zeta}{\sigma_{ij}}\frac{1+\delta}{c+\delta}}e^{-\frac{t^2}{2}}\,dt>1-\varepsilon \bigg\}. 
\end{equation} 
Let us denote by $\nu_{\varepsilon}>0$ the solution of $$\int_{-\nu_\varepsilon}^{\nu_{\varepsilon}}e^{-t^2/2}\,dt=\sqrt{2\pi}(1-\varepsilon),$$ a simple monotonicity argument suffices to explain the three branches of $\mathcal R^{(i,j)}$:

\noindent If $ \frac{\zeta}{\sigma_{ij}}\frac{1}{c}\geq \nu_{\varepsilon}$, then clearly the infimum $\delta>0$ that satisfies \eqref{eq: riskcond} is $\delta=0$. 

\noindent If $ \frac{\zeta}{\sigma_{ij}}\leq \nu_{\varepsilon}$, we have the other extreme, where no finite value of $\delta$ satisfies \eqref{eq: riskcond}. In this case, we have $\delta=+\infty$.

\noindent If $\frac{\zeta}{\sigma_{ij}}\frac{1}{c} < \nu_{\varepsilon}< \frac{\zeta}{\sigma_{ij}}$ there is unique $\delta^*>0$ such that $$\frac{\zeta}{\sigma_{ij}}\frac{1+\delta^*}{c+\delta^*}=\nu_\varepsilon.$$ The result follows after solving for $\delta^*$.  \hfill$\blacksquare$

\subsubsection{Proposition \ref{prop: limitsigma}}
Recall from properties of function $f(s;k)\in \cup_{l=1}^3 \mathbb W_l$ discussed in Appendix \ref{app: spectralfunctions}, which $f(s;k)$ is positive and finite. In addition, $f$ diverges on the boundary of $\bigcup_{l=1}^3\mathbb W_{l}$. Consequently, $f$ must attain a global minimum in the interior of $\bigcup_{l=1}^3\mathbb W_{l}$. With a little abuse of notation, the global minimum in the interior of $\bigcup_{l=1}^3\mathbb W_{l}$ will also reveal $\mathfrak f_l(k_1,k_2):=\mathfrak f(d\tau,\lambda_l\tau~;~k_1\tau^2,k_2\tau)$, as well as every finite (and positive) weighted sum of functions $\{\mathfrak f_l(k_1,k_2)\}_{l=2}^n$. Consequently, $\sigma_{ij}^2$ will attain a minimum over all $i\neq j$.\hfill$\blacksquare$

\subsubsection{Theorem \ref{thm: trade-off}}
Following the argument in \cite{somarakisnader_tac_1} and \cite{8884747}, one can calculate a common limit for the product of the systemic risk and the effective resistance. The trade-off equation \eqref{eq: tradeoff} can be derived similarly to risk/connectivity trade-off conditions detailed in \cite{somarakisnader_tac_1} and  \cite{8884747} for the first- or second-order consensus dynamics. \hfill$\blacksquare$

\subsection{Construction of sets $\mathbb W_i$,~$i=0,1,2,3$.}\label{appendix1} 
For the exposition of Theorem \ref{thm: stability}, we will define the sets that constitute the stability area of dynamics. To this end it is helpful to introduce notation $(s;k)$, for $s=(s_1,s_2)\in \mathbb R^2$ and $k=(k_1,k_2)\in \mathbb R^2$. Such distinction assists in separating arguments into $s_1,s_2$ (i.e. the systemic quantities in power system) and $k_1,k_2$ (i.e. the control quantities).
The first cut-off limit to be defined is
\begin{equation}\label{eq: bound1}
	\overline{\tau}_{s;k}=\frac{1}{\sqrt{k_2^2-s_1^2}}\text{arccot}\bigg(\frac{-s_1}{\sqrt{k_2^2-s_1^2}} \bigg),
\end{equation} 
whenever $k_2^2>s_1^2$. Next, for $\Delta_{s;k}=k_2^2+2s_2-s_1^2$, define $\gamma_{\pm}^2=\gamma_{\pm}^2(s;k)$ as  $$\gamma_{\pm}^2=\frac{1}{2}\bigg\{\Delta_{s;k} \pm \sqrt{\Delta_{s;k}^2-4(s_2^2-k_1^2)}  \bigg\}.$$ If $s_2^2\leq k_1^2$ there is only one positive solution $\gamma_{+}>0$. Otherwise if $s_2^2 > k_1^2$, there are two positive solutions $\gamma_+>\gamma_->0$. Whenever solutions are defined, the next quantity to be considered are the angles $\varphi_{\pm}=\varphi_{\pm}(s;k)\in [0,2\pi)$ via trigonometric numbers $$\cos(\varphi_{\pm})=-\frac{s_1 k_2 \gamma_{\pm}^2+k_1\big(s_2-\gamma_{\pm}^2\big)}{k_2^2 \gamma_\pm^2+k_1^2}, $$ $$ \sin(\varphi_{\pm})=\frac{s_1 k_1\gamma_{\pm}-k_2\gamma_\pm\big(s_2-\gamma_\pm^2\big)}{k_2^2 \gamma_\pm^2+k_1^2}. $$ 
Next, the critical cut-offs 
\begin{equation}\label{eq: bound2}
\tau_{+}^{(l)}=\frac{\varphi_{+}+2l\pi}{\gamma_{+}},~~~\tau_{-}^{(l)}=\frac{\varphi_{-}+2l\pi}{\gamma_{-}}, 
\end{equation} for $l=0,1,2,\dots$, are sorted as \begin{equation*}
\begin{split} 
0<&\tau_{+}^{(0)}<\tau_{-}^{(0)}<\tau_{+}^{(1)}<\tau_{-}^{(1)}<\dots \\ 
&\dots<\tau_{+}^{(l^*-1)}<\tau_{-}^{(l^*-1)}<\tau_{+}^{(l^*)}<\tau_{-}^{(l^*)}<\tau^{(l^*+1)}_+ 
\end{split}
\end{equation*} for the minimum positive integer $l^*$ that satisfies 
$$ \frac{\gamma_-(\varphi_++2\pi)-\varphi_- \gamma_+}{2\pi(\gamma_+-\gamma_-)}<l^*<\frac{(2\pi-\varphi_-)\gamma_+ +\varphi_+\gamma_-}{2\pi\big(\gamma_+ -\gamma_-\big)}.$$ 
The assortment of critical cut-offs signifies the transition of system \eqref{eq: linearizedwithfeedback} from stability to instability and back to stability. Under additional assumptions (explained in Theorem \ref{thm: stability}), the system is asymptotically stable for $\tau\in[0,\tau_+^{(0)})$, unstable for $\tau\in\big(\tau_+^{(0)},\tau_{-}^{(0)}\big)$, stable for $\tau\in \big(\tau_{-}^{(0)},\tau_{+}^{(1)}\big)$ and so forth, up until the last region of stability, that is when $\tau\in \big(\tau_{-}^{(l^*-1)},\tau_{+}^{(l^*)}\big)$ (see p. 77 of \cite{kuang1993delay}). The union of all stability regions constitute the set $\mathfrak{T}=\mathfrak{T}(s;k)$
\begin{equation}\label{eq: tauset}
\mathfrak{T}=\big[0,\tau_{}^{(0)}\big)\cup\bigg\{\bigcup_{l=1}^{l^*-1}\big(\tau_{-}^{(l)},\tau_{+}^{(l+1)}\big) \bigg\}.
\end{equation}
After implementing the transformation \eqref{eq: transform}\footnote{Here we have generic $s_2, k_1, k_2$ instead of $\lambda_j,\mu_j,\kappa_j$.}, the set \eqref{eq: tauset} is, in turn, transformed to set $\mathfrak{I}=\mathfrak{I}_{s;k}$  \begin{equation*}
\begin{split}
&\mathfrak{I}=\big\{\gamma_{+}<\varphi_{+}\big\}\bigcup  \\ 
&\bigg\{\bigcup_{l=1}^{l^*} \big\{\gamma_->\varphi_- + 2(l-1)\pi \big\}~ \cap ~\big\{ \gamma_+<\varphi_++2l\pi\big\} \bigg\}. 
\end{split}
\end{equation*} 
Based on the aforementioned quantities, we can proceed to define the sets $\mathbb W_i(s;k)\subset \mathbb R^{2}\times \mathbb R^{2}_+$  as in Table I.

\begin{prop}
Let $s=(s_1,s_2)\in \mathbb R_+^2$ be arbitrary but fixed. Union $\bigcup_{j=1}^3 \overline{\mathbb W}_{j}$ parametrized by $k=(k_1,k_2)\in \mathbb R^2$ contains the origin, it is compact and connected.
\end{prop}

\begin{proof} First, one can observe that the origin belongs to $\mathbb W_1$. 

\noindent{\underline{Compactness}}: It is straightforward to show that given $s_1,s_2>0$, the set $\overline{\mathbb W}_1$ is compact: $| k_1| \leq s_2$ and $|k_2|\leq \sqrt{s_1^2-2s_2+2\sqrt{k_2^2-s_1^2}} $. Set $\overline{\mathbb W}_2$ is also compact : $k_2\geq -s_1$ $k_1\geq -s_2$ and in the view of monotonicity of $\gamma_+(k_1,k_2)$, condition $\gamma_+\leq \varphi_+\in [0,2\pi]$ provides a sharp bound that is always achieved when $k_1$ and/or $k_2$ are large enough. Finally, the compactness of $\overline{\mathbb W}_3$ is proved using similar arguments: $|k_1|\leq s_2$ and $k_2\geq -s_1$. Now, as $k_2$ increases, $\gamma_+$ increases to $+\infty$, and $\gamma_->0$ vanishes monotonically. Elementary analysis reveals that for a fixed $k_1$, $\lim_{k_2\rightarrow +\infty}\varphi_{-}(k_1,k_2)\in [\pi, \frac{3\pi}{2}]$. Consequently, for any such fixed $k_1$ with $(k_1,k_2)~\in~\mathbb W_3$, the boundary $\gamma_-=\varphi_- + 2(l-1)\pi$ is achieved for $l\geq 1$ and finite $k_2$.

\noindent{\underline{Connectedness}}: Clearly, $\overline{\mathbb W}_1 \cap \overline{\mathbb W}_3 \neq \emptyset$. This concerns all points in $\mathbb W_1$ on the curve $k_2^2+2s_2-s_1^2=2\sqrt{s_2^2-k_1^2}$ that satisfies the last condition of $\mathbb W_3$. Similarly, $\overline{\mathbb W}_1 \cap \overline{\mathbb W}_2\neq \emptyset$. This follows after setting $k_1=s_2$, which necessarily yields $\gamma_+(s_2,k_2)=0$ for $(s_2,k_2)\in \overline{\mathbb W}_1$. 

\end{proof}

\subsection{The Spectral Function $f(s;k)$ }\label{app: spectralfunctions}
Function $f$, defined in \eqref{eq: f} in the form of improper integral, plays a central role in this work. Its arguments are presented as $(s;k)$, where  $s=(s_1,s_2)^T \in \mathbb R_+^2$ correspond to the fixed parameters of nominal system and $k=(k_1,k_2)^T\in \mathbb R^2$ corresponds to controllable parameters. Its domain is the union of sets  $\mathbb W_i(s;k)\subset \mathbb R^2\times \mathbb R_+^2$ and its range is $\mathbb R_+$.

\begin{prop} Function $f$ attains the properties:
\begin{enumerate}
\item For fixed $s\in \mathbb R_+^2$, $(s;k)\in \text{dom} f $ , $f(s;k)\geq\underline{f}>0$.
\item For a increasing sequence $\big\{s_1^{(j)}\big\}_{j\geq 1}$ with $\lim_{j} s_1^{(j)}\rightarrow + \infty$, and $(s_1^{(j)},s_2;k)\in \bigcup_{l=1}^3 \overline{\mathbb W}_{l} $ for $j\geq 1$, we have
$\lim_{j} f(s_1^{(j)},s_2;k)= 0$

\item For a increasing sequence $\big\{s_2^{(j)}\big\}_{j\geq 1}$ with $\lim_{j} s_2^{(j)}\rightarrow + \infty$, and $(s_1,s_2^{(j)};k)\in \bigcup_{l=1}^3 \overline{\mathbb W}_{l} $ for $j\geq 1$, we have $\lim_{j} f(s_1,s_2^{(j)};k)= 0$.
\end{enumerate}
\end{prop}


\begin{figure}
    \center
	\includegraphics[width=0.8\linewidth]{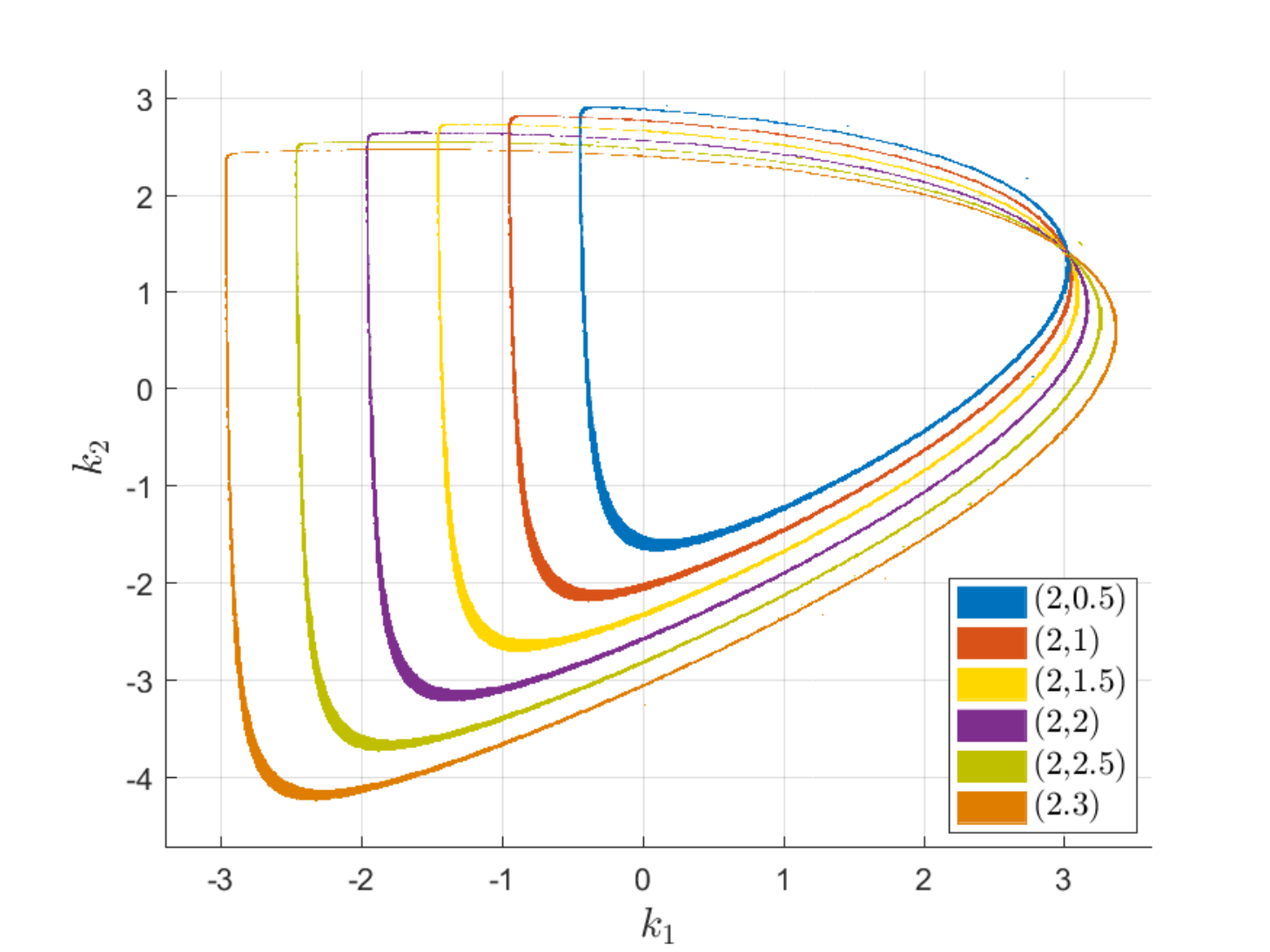}
	\caption{Stability areas for various selections $s=(s_1,s_2)$. The curves are constructed from boundaries of $\{\mathbb W_j\}_{j=1}^3$.}
\end{figure}

\begin{figure}
    \center
	\includegraphics[width=0.9\linewidth]{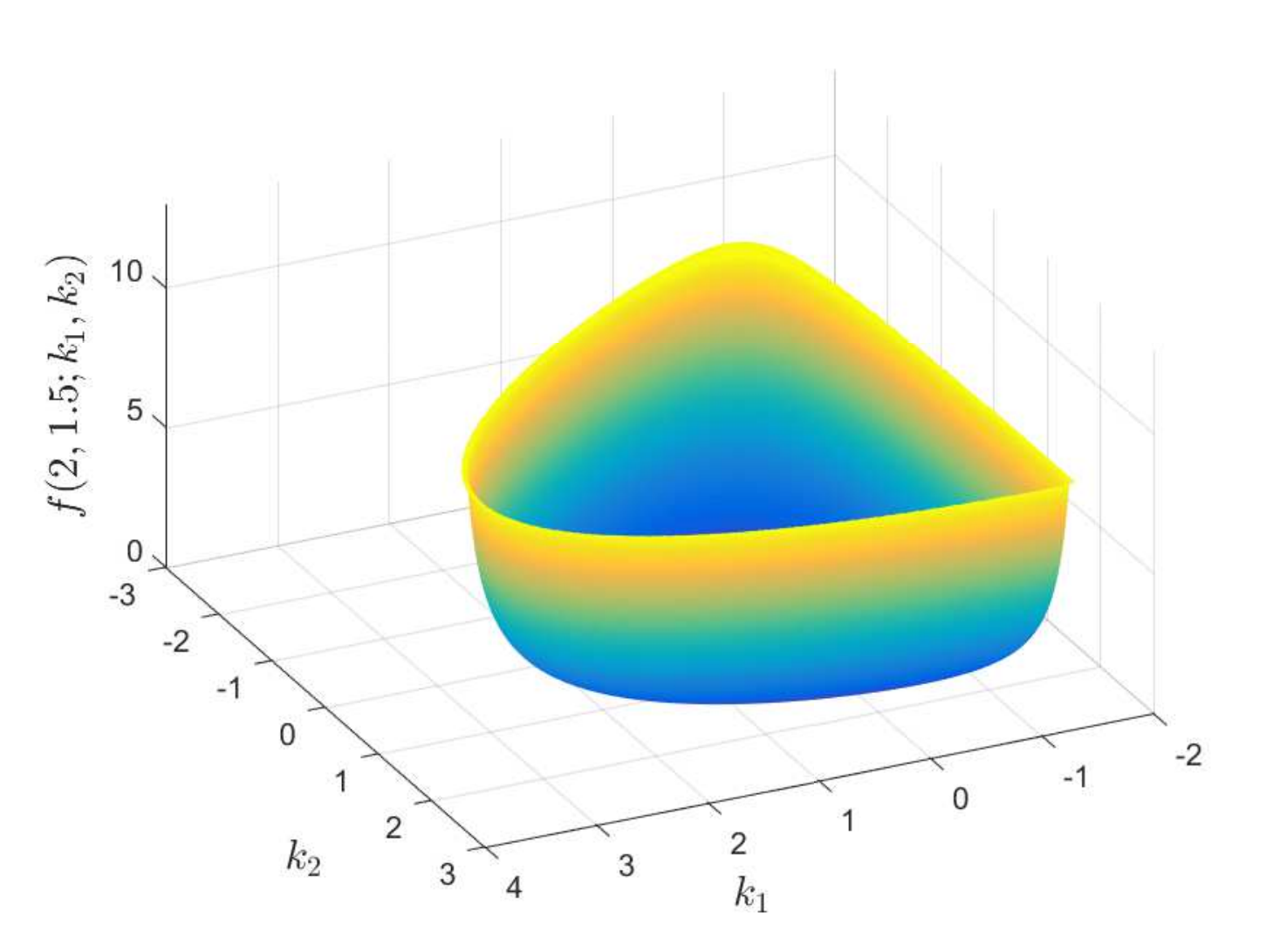}
	\caption{The graph of $f(2,1.5 ; k_1,k_2)$.}
\end{figure}

\begin{table*}\centering
\begin{minipage}{1.0\textwidth} \hrule \center
\begin{equation}\label{eq: f} f(s;k)=\int_{\mathbb R}\frac{ dr }{2\big( (s_1k_2-k_1)r^2+s_2k_1\big)\cos(r)-2r(k_2r^2+s_1k_1-k_2s_2)\sin(r)+r^4+(s_1^2+k_2^2-2s_2)r^2+s_2^2+k_1^2}
\end{equation}
\hrule
\end{minipage}
\end{table*}

\begin{proof}
\noindent (1): Recall that the denominator in the integrand in \eqref{eq: f} is the square of magnitude of characteristic function. It is, therefore, an even and non-negative function for every $r$. It is not hard to see that the denominator is bounded as follows 
\begin{equation*}
r^4+ \alpha_3 |r| r^2 + \alpha_2 r^2 +\alpha_1 |r| +\alpha_0 < \begin{cases}
 \beta_1 r   + \alpha_0 &,~~  r\in (0,1) \\
 \beta_2 r^4   &,~~ r>1
\end{cases}
\end{equation*}    
with $\alpha_3=2k_2$, $\alpha_2=2|s_1k_2-k_1|+|s_1^2+k_2^2-2s_2|$, $\alpha_1=2|s_1k_1-s_2k_2 |$, $\alpha_0=s_2^2+k_1^2$, $\beta_1=\sum_{i=1}^3 \alpha_i$, $\beta_2=\sum_{i=0}^3\alpha_i$, all as functions of $s,k$. Consequently,
\begin{equation*}
    \begin{split}
        f(s;k) & > \int_0^1 \frac{2\,dr}{\beta_1 r+\alpha_0} + \int_1^\infty \frac{2\,dr}{ \beta_2 r^4}\\
        & = \frac{2}{\beta_1}\ln\bigg(1+\frac{\alpha_0}{\beta_1} \bigg)+\frac{2}{3\beta_2}>0.
    \end{split}
\end{equation*}
\noindent (2): Denote by $g_r(s;k)$ the integrand of \eqref{eq: f}. It can be shown that $\frac{\partial}{\partial s_1}g_r(s;k)<0$ for  $s_1>|k_1|-k_2$. Choosing a large enough $j^*$, to end up with sequence $g^{(j)}(r):=g_r(s_1^{(j)},s_2;k)$ for $j>j^*(r)$, each of which is dominated by $g_{j^*}$. Also, $\lim_j g_j(r)\equiv 0$. The application of the dominated convergence theorem yields
\begin{equation*}
    \begin{split}
    \lim_{j} f(s_1^{(j)},s_2;k) & =\lim_j 2\int_{0}^\infty g_j(r)\,dr\\ 
    							& = 2\int_{0}^\infty \lim_j g_j(r)\,dr=0.
    \end{split}
\end{equation*}
\noindent (3): For any $\varepsilon>0$, pick any $M>\sqrt[3]{\frac{4}{3\varepsilon}}$ and choose $j^*$ large enough such that $s_2^{(j)}>\max\big\{ \sqrt[3]{\frac{4}{3\varepsilon}}-k_1,M^2-k_1-k_2\big\}$ for $j>j^{*}$. Fix such $j$ and observe that 
\begin{equation*}
\begin{split}
f(s_1,s_2^{(j)};k)& = 2 \int_{0}^M g(r)\,dr + 2 \int_M^\infty g(r)\,dr   \\
&<2 \int_{0}^M g(r)\,dr  +  2 \int_M^\infty \frac{dr}{r^4}\\
&< \frac{2M}{2(s_2^{(j)})k_1+(s_2^{(j)})^2+k_1^2}+ \frac{2}{3 M^3}\\& < \frac{\varepsilon}{2} + \frac{\varepsilon}{2}=\varepsilon.
\end{split}
\end{equation*}
The last step is due to the choices of $M$ and index $j$. For an arbitrarily small $\varepsilon>0$, the result follows.
\end{proof}

\bibliographystyle{IEEEtran}    
\bibliography{ifacconf}

\begin{thebibliography}{17}
\providecommand{\natexlab}[1]{#1}
\providecommand{\url}[1]{\texttt{#1}}
\providecommand{\urlprefix}{URL }
\expandafter\ifx\csname urlstyle\endcsname\relax
  \providecommand{\doi}[1]{doi:\discretionary{}{}{}#1}\else
  \providecommand{\doi}{doi:\discretionary{}{}{}\begingroup
  \urlstyle{rm}\Url}\fi

\bibitem[{Bamieh et~al.(2012)Bamieh, Jovanovi\'c, Mitra, and
  Patterson}]{Bamieh12}
Bamieh, B., Jovanovi\'c, M., Mitra, P., and Patterson, S. (2012).
\newblock Coherence in large-scale networks: Dimension-dependent limitations of
  local feedback.
\newblock \emph{Automatic Control, IEEE Transactions on}, 57(9), 2235 --2249.

\bibitem[{D\"{o}rfler and Bullo(2012)}]{doi:10.1137/110851584}
D\"{o}rfler, F. and Bullo, F. (2012).
\newblock Synchronization and transient stability in power networks and
  nonuniform kuramoto oscillators.
\newblock \emph{SIAM Journal on Control and Optimization}, 50(3), 1616--1642.

\bibitem[{Eliasson and Hill(1992)}]{141692}
Eliasson, B.E. and Hill, D.J. (1992).
\newblock Damping structure and sensitivity in the nordel power system.
\newblock \emph{IEEE Transactions on Power Systems}, 7(1), 97--105.

\bibitem[{Fardad et~al.(2014)Fardad, Lin, and Jovanovic}]{6716960}
Fardad, M., Lin, F., and Jovanovic, M.R. (2014).
\newblock Design of optimal sparse interconnection graphs for synchronization
  of oscillator networks.
\newblock \emph{IEEE Transactions on Automatic Control}, 59(9), 2457--2462.

\bibitem[{Mohammed(1984)}]{Mohammed84}
Mohammed, S.E.A. (1984).
\newblock \emph{Stochastic Functional Differential Equations}.
\newblock Pilman Advanced Publishing.

\bibitem[{Paganini and Mallada(2017)}]{8262755}
Paganini, F. and Mallada, E. (2017).
\newblock Global performance metrics for synchronization of heterogeneously
  rated power systems: The role of machine models and inertia.
\newblock In \emph{2017 55th Annual Allerton Conference}, 324--331.

\bibitem[{Poolla et~al.(2017)Poolla, Bolognani, and D\"{o}rfler}]{7924418}
Poolla, B.K., Bolognani, S., and D\"{o}rfler, F. (2017).
\newblock Optimal placement of virtual inertia in power grids.
\newblock \emph{IEEE Transactions on Automatic Control}, 62(12), 6209--6220.

\bibitem[{Sauer and Pai(2006)}]{sauer2006power}
Sauer, P. and Pai, M. (2006).
\newblock \emph{Power System Dynamics and Stability}.
\newblock Stipes Publishing L.L.C.

\bibitem[{Siami and Motee(2016)}]{Siami16TACa}
Siami, M. and Motee, N. (2016).
\newblock Fundamental limits and tradeoffs on disturbance propagation in linear
  dynamical networks.
\newblock \emph{IEEE Transactions on Automatic Control}, 61(12), 4055--5062.

\bibitem[{Somarakis et~al.(2017)Somarakis, Ghaedsharaf, and
  Motee}]{somyasnader17}
Somarakis, C., Ghaedsharaf, Y., and Motee, N. (2017).
\newblock Aggregate fluctuations in time-delay linear consensus networks: A
  systemic risk perspective.
\newblock In \emph{The 2017 American Control Conference}.

\bibitem[{Somarakis et~al.(2018{\natexlab{a}})Somarakis, Ghaedsharaf, and
  Motee}]{somyasnader18}
Somarakis, C., Ghaedsharaf, Y., and Motee, N. (2018{\natexlab{a}}).
\newblock Risk of collision in a vehicle platoon in presence of communication
  time delay and exogenous stochastic disturbance.
\newblock In \emph{The 57th IEEE CDC}.
\newblock (to appear).

\bibitem[{Somarakis et~al.(2018{\natexlab{b}})Somarakis, Ghaedsharaf, and
  Motee}]{DBLP:journals/corr/abs-1801-06856}
Somarakis, C., Ghaedsharaf, Y., and Motee, N. (2018{\natexlab{b}}).
\newblock Time-delay origins of fundamental tradeoffs between risk of large
  fluctuations and network connectivity.
\newblock \emph{CoRR}, abs/1801.06856.
\newblock \urlprefix\url{http://arxiv.org/abs/1801.06856}.

\bibitem[{Summers et~al.(2015)Summers, Shames, Lygeros, and Dörfler}]{7330605}
Summers, T., Shames, I., Lygeros, J., and Dörfler, F. (2015).
\newblock Topology design for optimal network coherence.
\newblock In \emph{2015 European Control Conference (ECC)}, 575--580.

\bibitem[{Tegling et~al.(2015)Tegling, Bamieh, and Gayme}]{7086037}
Tegling, E., Bamieh, B., and Gayme, D.F. (2015).
\newblock The price of synchrony: Evaluating the resistive losses in
  synchronizing power networks.
\newblock \emph{IEEE Transactions on Control of Network Systems}, 2(3),
  254--266.

\bibitem[{Tong(1990)}]{tong90}
Tong, Y.L. (1990).
\newblock \emph{The Multivariate Normal Distribution}.
\newblock Springer, 1st edition.

\bibitem[{Varaiya et~al.(1985)Varaiya, Wu, and Chen}]{1457634}
Varaiya, P., Wu, F.F., and Chen, R.L. (1985).
\newblock Direct methods for transient stability analysis of power systems:
  Recent results.
\newblock \emph{Proceedings of the IEEE}, 73(12), 1703--1715.

\bibitem[{Vu and Turitsyn(2017)}]{Vu2017}
Vu, L.T. and Turitsyn, K. (2017).
\newblock A framework for robust assessment of power grid stability and
  resiliency.
\newblock \emph{IEEE Transactions on Automatic Control}, 62(3), 1165--1177.

\end{thebibliography}


\begin{thebibliography}{10}
\providecommand{\url}[1]{#1}
\csname url@samestyle\endcsname
\providecommand{\newblock}{\relax}
\providecommand{\bibinfo}[2]{#2}
\providecommand{\BIBentrySTDinterwordspacing}{\spaceskip=0pt\relax}
\providecommand{\BIBentryALTinterwordstretchfactor}{4}
\providecommand{\BIBentryALTinterwordspacing}{\spaceskip=\fontdimen2\font plus
\BIBentryALTinterwordstretchfactor\fontdimen3\font minus
  \fontdimen4\font\relax}
\providecommand{\BIBforeignlanguage}[2]{{%
\expandafter\ifx\csname l@#1\endcsname\relax
\typeout{** WARNING: IEEEtran.bst: No hyphenation pattern has been}%
\typeout{** loaded for the language `#1'. Using the pattern for}%
\typeout{** the default language instead.}%
\else
\language=\csname l@#1\endcsname
\fi
#2}}
\providecommand{\BIBdecl}{\relax}
\BIBdecl

\bibitem{fox2010smart}
\BIBentryALTinterwordspacing
P.~Fox-Penner, \emph{Smart Power: Climate Change, the Smart Grid, and the
  Future of Electric Utilities}.\hskip 1em plus 0.5em minus 0.4em\relax Island
  Press, 2010. [Online]. Available:
  \url{https://books.google.com/books?id=WztHxkT36LoC}
\BIBentrySTDinterwordspacing

\bibitem{momoh2012smart}
\BIBentryALTinterwordspacing
J.~Momoh, \emph{Smart Grid: Fundamentals of Design and Analysis}, ser. I E E
  Power Engineering Series.\hskip 1em plus 0.5em minus 0.4em\relax Wiley, 2012.
  [Online]. Available: \url{https://books.google.com/books?id=G3prlp3jD4QC}
\BIBentrySTDinterwordspacing

\bibitem{blackout1996}
V.~V. amd Y.~Li, ``Analysis of the 1996 western american electric blackouts,''
  in \emph{Bulk Power System Dynamics and Control - VI}, June 2004.

\bibitem{blackout2003_2}
U.-C. P. S. O.~T. Force, ``Final report on the august 14, 2003 blackout in the
  united states and canada: Causes and recommendations,'' Tech. Rep., April
  2004.

\bibitem{blackoutitaly_1}
S.~Corsi and C.~Sabelli, ``General blackout in italy sunday september 28,
  2003,'' in \emph{Power Engineering Society General Meeting}, vol.~2, June
  2004, pp. 1691 -- 1702.

\bibitem{1705631}
F.~{Blaabjerg}, R.~{Teodorescu}, M.~{Liserre}, and A.~V. {Timbus}, ``Overview
  of control and grid synchronization for distributed power generation
  systems,'' \emph{IEEE Transactions on Industrial Electronics}, vol.~53,
  no.~5, pp. 1398--1409, Oct 2006.

\bibitem{5454394}
F.~{Rahimi} and A.~{Ipakchi}, ``Demand response as a market resource under the
  smart grid paradigm,'' \emph{IEEE Transactions on Smart Grid}, vol.~1, no.~1,
  pp. 82--88, June 2010.

\bibitem{sauer2006power}
P.~Sauer and M.~Pai, \emph{Power System Dynamics and Stability}.\hskip 1em plus
  0.5em minus 0.4em\relax Stipes Publishing L.L.C., 2006.

\bibitem{1457634}
P.~Varaiya, F.~F. Wu, and R.-L. Chen, ``Direct methods for transient stability
  analysis of power systems: Recent results,'' \emph{Proceedings of the IEEE},
  vol.~73, no.~12, pp. 1703--1715, Dec 1985.

\bibitem{1994power}
P.~Kundur, \emph{Power System Stability And Control}, ser. McGraw-Hill
  Publishing Company.\hskip 1em plus 0.5em minus 0.4em\relax McGraw-Hill, 1994.

\bibitem{260906}
R.~{Grondin}, I.~{Kamwa}, L.~{Soulieres}, J.~{Potvin}, and R.~{Champagne}, ``An
  approach to pss design for transient stability improvement through
  supplementary damping of the common low-frequency,'' \emph{IEEE Transactions
  on Power Systems}, vol.~8, no.~3, pp. 954--963, Aug 1993.

\bibitem{1709097}
I.~{Kamwa}, J.~{Beland}, G.~{Trudel}, R.~{Grondin}, C.~{Lafond}, and
  D.~{McNabb}, ``Wide-area monitoring and control at hydro-quebec: past,
  present and future,'' in \emph{2006 IEEE Power Engineering Society General
  Meeting}, June 2006, pp. 12 pp.--.

\bibitem{6740090}
F.~{D\"{o}rfler}, M.~R. {Jovanović}, M.~{Chertkov}, and F.~{Bullo},
  ``Sparsity-promoting optimal wide-area control of power networks,''
  \emph{IEEE Transactions on Power Systems}, vol.~29, no.~5, pp. 2281--2291,
  Sep. 2014.

\bibitem{6580901}
A.~{Chakrabortty} and P.~P. {Khargonekar}, ``Introduction to wide-area control
  of power systems,'' in \emph{2013 American Control Conference}, June 2013,
  pp. 6758--6770.

\bibitem{5728885}
S.~Wang, X.~Meng, and T.~Chen, ``Wide-area control of power systems through
  delayed network communication,'' \emph{IEEE Transactions on Control Systems
  Technology}, vol.~20, no.~2, pp. 495--503, March 2012.

\bibitem{doi:10.1002/etep.545}
Y.~Chompoobutrgool, L.~Vanfretti, and M.~Ghandhari, ``Survey on power system
  stabilizers control and their prospective applications for power system
  damping using synchrophasor-based wide-area systems,'' \emph{European
  Transactions on Electrical Power}, vol.~21, no.~8, pp. 2098--2111, 2011.

\bibitem{Annaswamy2013}
\BIBentryALTinterwordspacing
A.~M. Annaswamy, D.~Soudbakhsh, R.~Schneider, D.~Goswami, and S.~Chakraborty,
  \emph{Arbitrated Network Control Systems: A Co-Design of Control and Platform
  for Cyber-Physical Systems}.\hskip 1em plus 0.5em minus 0.4em\relax
  Heidelberg: Springer International Publishing, 2013, pp. 339--356. [Online].
  Available: \url{https://doi.org/10.1007/978-3-319-01159-2_18}
\BIBentrySTDinterwordspacing

\bibitem{bullo2018lectures}
\BIBentryALTinterwordspacing
F.~Bullo, \emph{Lectures on Network Systems}.\hskip 1em plus 0.5em minus
  0.4em\relax CreateSpace Independent Publishing Platform, 2018. [Online].
  Available: \url{https://books.google.com/books?id=JpoxtwEACAAJ}
\BIBentrySTDinterwordspacing

\bibitem{SOUDBAKHSH2017171}
D.~Soudbakhsh, A.~Chakrabortty, and A.~M. Annaswamy, ``A delay-aware
  cyber-physical architecture for wide-area control of power systems,''
  \emph{Control Engineering Practice}, vol.~60, pp. 171 -- 182, 2017.

\bibitem{SCHIFFER2017261}
J.~Schiffer, F.~D\"{o}rfler, and E.~Fridman, ``Robustness of distributed
  averaging control in power systems: Time delays \& dynamic communication
  topology,'' \emph{Automatica}, vol.~80, pp. 261--271, 2017.

\bibitem{Vu2017}
L.~T. Vu and K.~Turitsyn, ``A framework for robust assessment of power grid
  stability and resiliency,'' \emph{IEEE Transactions on Automatic Control},
  vol.~62, no.~3, pp. 1165--1177, Mar 2017.

\bibitem{985276}
G.~T. {Heydt}, ``The impact of time delay on robust control design in power
  systems,'' in \emph{2002 IEEE Power Engineering Society Winter Meeting.
  Conference Proceedings}, vol.~2, Jan 2002, pp. 1511--1516 vol.2.

\bibitem{5688286}
N.~R. {Chaudhuri}, D.~{Chakraborty}, and B.~{Chaudhuri}, ``Damping control in
  power systems under constrained communication bandwidth: A predictor
  corrector strategy,'' \emph{IEEE Transactions on Control Systems Technology},
  vol.~20, no.~1, pp. 223--231, Jan 2012.

\bibitem{6517501}
S.~{Zhang} and V.~{Vittal}, ``Design of wide-area power system damping
  controllers resilient to communication failures,'' \emph{IEEE Transactions on
  Power Systems}, vol.~28, no.~4, pp. 4292--4300, Nov 2013.

\bibitem{follmer11}
H.~F\"{o}llmer and A.~Schied, \emph{Stochastic Finance. An introduction in
  discrete time.}\hskip 1em plus 0.5em minus 0.4em\relax De Gruyter, 2011.

\bibitem{somarakisnader_tac_1}
C.~Somarakis, Y.~Ghaedsharaf, and N.~Motee, ``Time-delay origins of fundamental
  tradeoffs between risk of large fluctuations and network connectivity,''
  \emph{IEEE Transactions on Automatic Control}, vol.~64, no.~9, Sept 2019.

\bibitem{8884747}
C.~{Somarakis}, Y.~{Ghaedsharaf}, and N.~{Motee}, ``Risk of collision and
  detachment in vehicle platooning: Time-delay–induced limitations and
  tradeoffs,'' \emph{IEEE Transactions on Automatic Control}, vol.~65, no.~8,
  pp. 3544--3559, 2020.

\bibitem{6759860}
M.~Siami and N.~Motee, ``Fundamental limits on robustness measures in networks
  of interconnected systems,'' in \emph{52nd IEEE Conference on Decision and
  Control}, 2013, pp. 67--72.

\bibitem{7086037}
E.~Tegling, B.~Bamieh, and D.~F. Gayme, ``The price of synchrony: Evaluating
  the resistive losses in synchronizing power networks,'' \emph{IEEE
  Transactions on Control of Network Systems}, vol.~2, no.~3, pp. 254--266,
  Sept 2015.

\bibitem{SOMARAKIS2018142}
\BIBentryALTinterwordspacing
C.~Somarakis, Y.~Ghaedsharaf, F.~Dörfler, and N.~Motee, ``Risk of phase
  incoherence in noisy power networks with delayed feedback control,''
  \emph{IFAC-PapersOnLine}, vol.~51, no.~23, pp. 142 -- 147, 2018, 7th IFAC
  Workshop on Distributed Estimation and Control in Networked Systems NECSYS
  2018. [Online]. Available:
  \url{http://www.sciencedirect.com/science/article/pii/S2405896318335572}
\BIBentrySTDinterwordspacing

\bibitem{somnader20}
C.~Somarakis and N.~Motee, ``Delay-aware risk analysis and control in smart
  grid networks with corrupted measurements,'' in \emph{The 2020 American
  Control Conference}, 2020.

\bibitem{7924418}
B.~K. Poolla, S.~Bolognani, and F.~D\"{o}rfler, ``Optimal placement of virtual
  inertia in power grids,'' \emph{IEEE Transactions on Automatic Control},
  vol.~62, no.~12, pp. 6209--6220, Dec 2017.

\bibitem{doi:10.1137/110851584}
F.~D\"{o}rfler and F.~Bullo, ``Synchronization and transient stability in power
  networks and nonuniform kuramoto oscillators,'' \emph{SIAM Journal on Control
  and Optimization}, vol.~50, no.~3, pp. 1616--1642, 2012.

\bibitem{Stallings:2006:DCC:1215307}
W.~Stallings, \emph{Data and Computer Communications (8th Edition)}.\hskip 1em
  plus 0.5em minus 0.4em\relax Upper Saddle River, NJ, USA: Prentice-Hall,
  Inc., 2006.

\bibitem{DBLP:journals/corr/abs-1711-10348}
\BIBentryALTinterwordspacing
T.~Coletta and P.~Jacquod, ``Performance measures in electric power networks
  under line contingencies,'' \emph{CoRR}, vol. abs/1711.10348, 2017. [Online].
  Available: \url{http://arxiv.org/abs/1711.10348}
\BIBentrySTDinterwordspacing

\bibitem{Siami16TACa}
M.~Siami and N.~Motee, ``Fundamental limits and tradeoffs on disturbance
  propagation in linear dynamical networks,'' \emph{IEEE Transactions on
  Automatic Control}, vol.~61, no.~12, pp. 4055--5062, 2016.

\bibitem{Horn:2012:MA:2422911}
A.~Horn, R and C.~R. Johnson, \emph{Matrix Analysis}, 1st~ed.\hskip 1em plus
  0.5em minus 0.4em\relax New York, NY, USA: Cambridge University Press, 2010.

\bibitem{gu2003stability}
\BIBentryALTinterwordspacing
K.~Gu, J.~Chen, and V.~Kharitonov, \emph{Stability of Time-Delay Systems}, ser.
  Control Engineering.\hskip 1em plus 0.5em minus 0.4em\relax Birkh{\"a}user
  Boston, 2003. [Online]. Available:
  \url{https://books.google.com/books?id=YzoQJHJpw\_4C}
\BIBentrySTDinterwordspacing

\bibitem{lunel93}
J.~K. Hale and S.~Lunel, \emph{Introduction to Functional Differential
  Equations}.\hskip 1em plus 0.5em minus 0.4em\relax Springer-Verlag, 1993,
  vol.~99.

\bibitem{Mohammed84}
S.-E.~A. Mohammed, \emph{Stochastic Functional Differential Equations}.\hskip
  1em plus 0.5em minus 0.4em\relax Pilman Advanced Publishing, 1984.

\bibitem{Rockafellar_Royset_2015}
R.~T. Rockafellar and J.~O. Royset, ``Risk measures in engineering design under
  uncertainty,'' in \emph{Proc. International Conf. on Applications of
  Statistics and Probability in Civil Engineering}, 2015.

\bibitem{klein1993}
D.~J. Klein and M.~Randi{\'{c}}, ``Resistance distance,'' \emph{Journal of
  Mathematical Chemistry}, vol.~12, no.~1, pp. 81--95, 1993.

\bibitem{pai1989}
M.~Pai, \emph{Energy Function Analysis For Power System Stability}.\hskip 1em
  plus 0.5em minus 0.4em\relax Kluwer Academic Publishers, 1989.

\bibitem{athay79}
T.~Athay, R.~Podmore, and S.~Virmani, ``A practical method for the direct
  analysis of transient stability,'' \emph{IEEE Transactions on Power Apparatus
  and Systems}, vol. PAS-98, no.~2, pp. 573 -- 584, March/April 1979.

\bibitem{kuang1993delay}
\BIBentryALTinterwordspacing
Y.~Kuang, \emph{Delay Differential Equations: With Applications in Population
  Dynamics}, ser. Mathematics in Science and Engineering.\hskip 1em plus 0.5em
  minus 0.4em\relax Elsevier Science, 1993. [Online]. Available:
  \url{https://books.google.com/books?id=YWsz1IXkvy0C}
\BIBentrySTDinterwordspacing

\end{thebibliography}

 \begin{IEEEbiography}[{\includegraphics[width=1in,height=1.25in,clip,keepaspectratio]{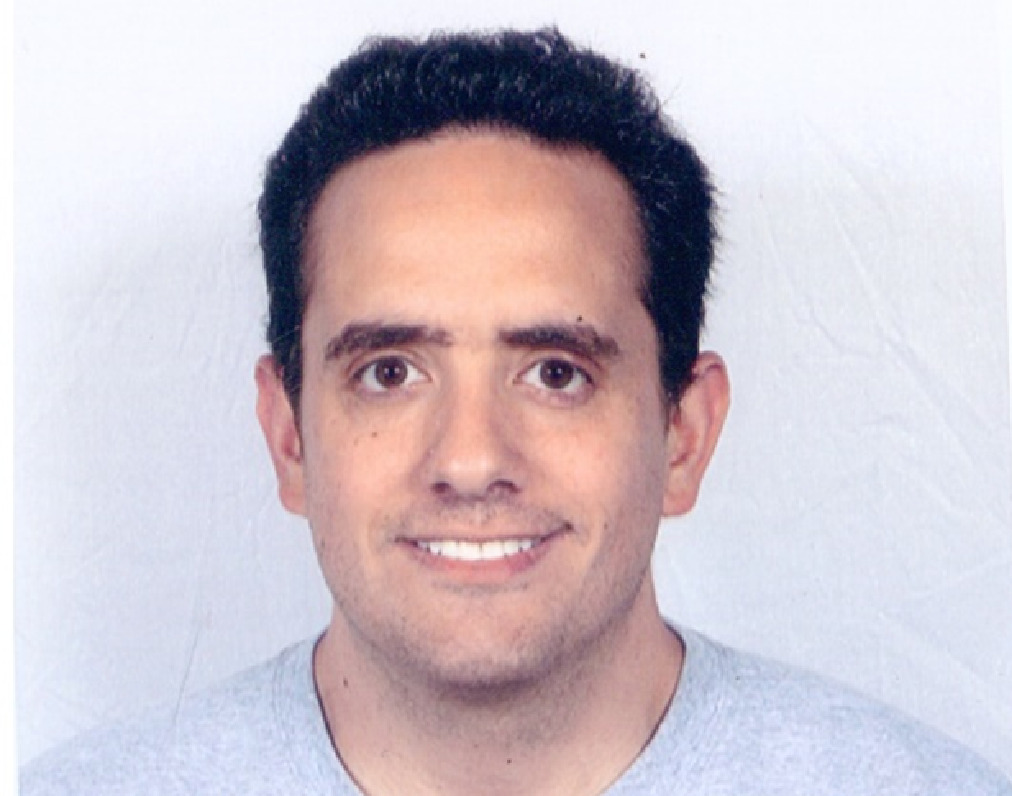}}]{Christoforos Somarakis}
Christoforos Somarakis received the B.S. degree
in Electrical Engineering from the National Technical
University of Athens, Athens, Greece, in 2007
and the M.S. and Ph.D. degrees in applied
mathematics from the University of Maryland at
College Park, in 2012 and 2015, respectively. He
was a Post-Doctoral scholar and a Research Scientist
with the Department of Mechanical Engineering and
Mechanics at Lehigh University from 2016 to 2019.
He is currently member of research staff with the
System Sciences Lab at Palo Alto Research Center.
 \end{IEEEbiography}

\begin{IEEEbiography}[{\includegraphics[width=1in,height=1.25in,clip,keepaspectratio]{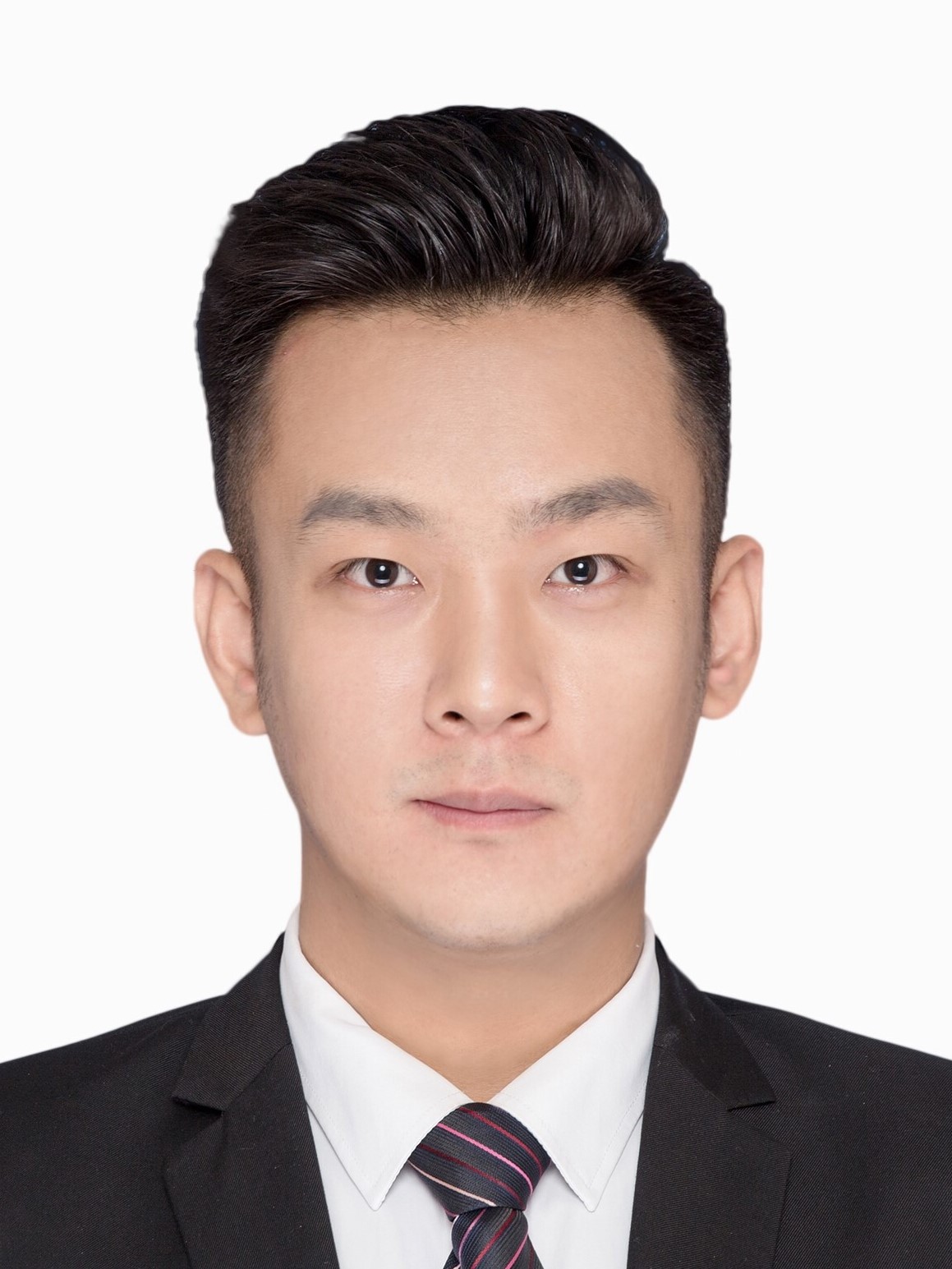}}]{Guangyi Liu}
    Guangyi Liu received his B.E. degree in
    aircraft design and engineering from the Beijing Institute of Technology in 2016 and the M.S. degree in mechanical engineering from the Lehigh University in 2018. He is currently pursuing a Ph.D. degree
    in the Department of Mechanical Engineering and
    Mechanics at Lehigh University. His research interests include the risk analysis of networked control systems and multi-agent machine learning.
\end{IEEEbiography}
 
 \begin{IEEEbiography}[{\includegraphics[width=1in,height=1.25in,clip,keepaspectratio]{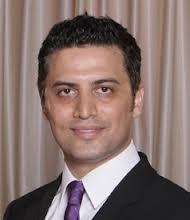}}]{Nader Motee}
 Nader Motee (Senior Member, IEEE) received 
the B.Sc. degree in electrical engineering from 
the Sharif University of Technology, Tehran, 
Iran, in 2000, and the M.Sc. and Ph.D. degrees 
in electrical and systems engineering from the 
University of Pennsylvania, Philadelphia, PA, 
USA, in 2006 and 2007, respectively. 
From 2008 to 2011, he was a Postdoctoral 
Scholar with the Control and Dynamical Systems Department, California Institute of Technology, Pasadena, CA, USA. He is currently a 
Professor with the Department of Mechanical Engineering and Mechanics, Lehigh University, Bethlehem, PA, USA. His research interests
include distributed control systems and real-time robot perception. 
Dr. Motee was the recipient of several awards including the 2019 Best 
SIAM Journal of Control and Optimization Paper Prize, the 2008 AACC 
Hugo Schuck Best Paper Award, the 2007 ACC Best Student Paper 
Award, the 2008 Joseph and Rosaline Wolf Best Thesis Award, the 
2013 Air Force Office of Scientific Research Young Investigator Program 
Award, 2015 NSF Faculty Early Career Development Award, and a 2016 
Office of Naval Research Young Investigator Program Award.
 \end{IEEEbiography}

\end{document}